\documentclass[12pt]{amsart}  

\usepackage{bbold}
\usepackage{amscd} 
\usepackage{amsmath,amssymb}
\usepackage{mathtools}
\usepackage{mathrsfs}
\usepackage{color}

\newtheorem{theorem}{Theorem}[section]
\newtheorem{lemma}[theorem]{Lemma}
\newtheorem{corollary}[theorem]{Corollary} 
\theoremstyle{definition}

\newtheorem{remark}[theorem]{Remark} 
\numberwithin{equation}{section}

\numberwithin{equation}{section}

\def\B{\mathscr B}

\def\D{\text{\rm dom}}
\def\dom{\text{\rm dom}}
\def\ran{\text{\rm ran}}

\def\RE{\mathbb R}
\def\CO{{\mathbb C}}

\def\ph*{\phi_\star}

\def\be{\begin{equation}}
\def\ee{\end{equation}}
\def\min{{\rm min}}
\def\max{{\rm max}}

\def\-{{\rm in}}
\def\+{{\rm ex}}

\def\bou{{\mathscr B}}

\def\ve{\varepsilon}

\title[The point scatterer approximation for wave dynamics 
]{The point scatterer approximation for wave dynamics 
}

\author{Andrea Mantile}
\author{Andrea Posilicano}
\address{Laboratoire de Math\'{e}matiques de Reims, UMR9008 CNRS et
Universit\'{e} de Reims Champagne-Ardenne, Moulin de la Housse BP 1039, 51687
Reims, France}
\address{DiSAT, Sezione di Matematica, Universit\`a dell'Insubria, via Valleggio 11, I-22100
Como, Italy}
\email{andrea.mantile@univ-reims.fr}
\email{andrea.posilicano@unisubria.it}

\begin{document}

\begin{abstract}
Given an open, bounded and connected set $\Omega\subset\mathbb{R}^{3}$ and
its rescaling $\Omega_{\varepsilon}$ of size $\varepsilon\ll 1$, we consider
the solutions of the Cauchy problem for the inhomogeneous wave equation 
$$
(\varepsilon^{-2}\chi_{\Omega_{\varepsilon}}+\chi_{\mathbb{R}^{3}\backslash\Omega_{\varepsilon}})\partial_{tt}u=\Delta u+f
$$
with initial data and source supported
outside $\Omega_{\varepsilon}$; here, 
$\chi_{S}$ denotes the characteristic function of a set $S$. We provide
the first-order $\varepsilon$-corrections with respect to the solutions of
the inhomogeneous free wave equation and give space-time estimates on the
remainders in the $L^{\infty}((0,1/\varepsilon^{\tau}),L^{2}(\mathbb{R}^{3}))
$-norm. Such corrections are explicitly expressed in terms of the
eigenvalues and eigenfunctions of the Newton potential operator in $L^{2}(\Omega)$ and provide an effective dynamics describing a legitimate
point scatterer approximation in the time domain.
\end{abstract}

\maketitle

\section{Introduction} Let $\Omega\subset\mathbb{R}^{3}$ be open, bounded and connected, $%
y_{0}\in\Omega$ and let $\varepsilon\in(0,1) $ be a small-scale
parameter. Denoting with $\Omega_{\varepsilon}$ the rescaled domain%
\begin{equation}
\Omega_{\varepsilon}:=\left\{ y=y_{0}+\varepsilon(x-y_{0})\,,\ x\in
\Omega\right\},   \label{Omega_eps}
\end{equation}
we introduce the contrast function%
\begin{equation*}
\varepsilon^{-2}\chi_{\Omega_{\varepsilon}}+\chi_{\mathbb{R}%
^{3}\backslash\Omega_{\varepsilon}}\,, 
\end{equation*}%
modeling a sharp discontinuity of a medium across the interface $\partial
\Omega_{\varepsilon}$. The perturbed wave equation%
\begin{equation}  \label{pwe}
(\varepsilon^{-2}\chi_{\Omega_{\varepsilon}}+\chi_{\mathbb{R}%
^{3}\backslash\Omega_{\varepsilon}})\partial_{tt}u=\Delta u\,,
\end{equation}
describes the interaction between a scalar wave and a small inhomogeneity
having high contrast with respect to an homogeneous background. Here, the
relative speed of propagation inside $\Omega_{\varepsilon}$ is given by the
number $\varepsilon$; hence, small values of $\varepsilon$ correspond to a
local regime of small relative speed of propagation. We are interested in
the asymptotic behavior of the solutions of the Cauchy problem for (\ref{pwe}%
) as $\varepsilon\searrow 0$.
\par
Under the time-harmonic assumption $u(t,x)=e^{it\omega}u_{\omega}(x)$, the
corresponding stationary equation writes as%
\begin{equation}
\big( \Delta+(\varepsilon^{-2}\chi_{\Omega_{\varepsilon}}+\chi_{\mathbb{R}%
^{3}\backslash\Omega_{\varepsilon}})\,\omega^{2}\big) u_{\omega}=0\,. 
\label{WE_eps_omega}
\end{equation}
It is known that, depending on the incident frequency $\omega$, the
scattering solutions $u^{sc}_{\omega} $ of (\ref{WE_eps_omega}) may
concentrate around $\Omega_{\varepsilon}$ as $\varepsilon\searrow0$. Namely,
specific values of $\omega$, referred to as \emph{resonant frequencies}, are
associated to a scattering enhancement localized at $y_{0}$. At these
specific frequencies, the dominant part of the scattered field assumes the
form%
\begin{equation}
u_{\omega}^{sc}( x) \sim u_{\omega}^{in}( y_{0})\, \Lambda_{\Omega}( \omega)
\,\frac{e^{-i\omega|x-y_{0}|}}{4\pi|x-y_{0}|}\,,   \label{point_scatterer}
\end{equation}
where $\Lambda_{\Omega}(\omega) $ is a scattering coefficient depending on
the physical parameters. Such asymptotic, provided in \cite{Amm1} under
far-field approximations, describes a diffusion from a \emph{point scatterer}
placed at $y_{0}$.
\par
The point scatterer model, introduced by Foldy in \cite{Foldy} and further
developed by Lax, \cite{Lax}, consists in using a (\ref{point_scatterer}%
)-like formula as an heuristic approximation for the stationary scattering
from a small obstacles placed at $y_{0}$. Such an approximation is commonly
adopted to describe scattering from small Dirichlet obstacles. An
interpretation, discussed in \cite{HuMaSi}, shows that the Foldy-Lax
description of a point scatterer corresponds to modeling the interaction
between the wave and the scatterer in terms of a point perturbation of the
Laplacian, a class of singular perturbations studied in \cite{AGHKH}, with a
frequency-dependent scattering length. It is worth to remark that such model has no direct counterpart in the time domain setting. However, since the point scatterer approximation
provides an effective limit for the asymptotic regime at specific
frequencies, the asymptotic analysis of the wave dynamics generated by
solutions of \eqref{pwe} at small scale may suggests a physically coherent
definition of a point scatterer in the time domain.
\par
The analysis of the stationary problem (\ref{WE_eps_omega}) in the small
scale limit $\varepsilon \searrow 0$ enlightens the emergency of a discrete
set of \emph{scattering resonances} at resonant frequencies $\omega _{k}$
related to the inverse spectral points of the \emph{Newton potential operator%
} 
\begin{equation*}
N_{0}:L^{2}(\Omega )\rightarrow L^{2}(\Omega )\,,\qquad N_{0}u(x):=\frac{1}{%
4\pi }\int_{\Omega }\frac{u(y)\,dy}{|x-y|}
\end{equation*}%
by $\omega _{k}=\lambda _{k}^{-1/2}$, where $\lambda _{k}\in\sigma _{d}(N_{0})$, see \cite{Amm1}, \cite{DaGaSi}. This suggests that
the generator of the dynamics may have spectral resonances at the energies $%
\lambda _{k}^{-1}$. Such a spectral problem is considered in \cite{MaPo 24}%
, where this picture is validated and precise estimates for the imaginary
parts of resonances are provided. In particular, under suitable assumptions
on $\Omega $, it is shown that at each $\lambda _{k}\in \sigma _{d}(N_{0})$
there corresponds a unique resonance of $-(\varepsilon ^{2}\chi _{\Omega
_{\varepsilon }}+\chi _{\mathbb{R}^{3}\backslash \Omega _{\varepsilon
}})\Delta $ converging to $\lambda _{k}^{-1}$ with imaginary part of order $%
\varepsilon $.
\par
According to the theory of second-order Cauchy problems with self-adjoint
generators, the solutions of (\ref{pwe}) express in terms of time
propagators which are related to the inverse Laplace transform of the
resolvent operator $(-(\varepsilon^{2}\chi_{\Omega_{\varepsilon}}+
\chi_{\mathbb{R}^{3}\backslash\Omega_{\varepsilon}})\Delta+z^{2})^{-1}$, see, e.g., \cite[Sec I.3.14]{Arendt}. Hence, the relevant properties of
the dynamical system are encoded in the spectral profile of $%
(\varepsilon^{2}\chi_{\Omega_{\varepsilon}}+\chi_{\mathbb{R}%
^{3}\backslash\Omega_{\varepsilon}})\Delta$, including eigenvalues and
resonances. In view of the results from \cite{MaPo 24}, we expect that in
the small-scale regime, the wave dynamics (\ref{pwe}) may be governed by a
(possibly finite) number of resonant states whose survival time, defined by
the imaginary part of spectral resonances, is of order $1/\varepsilon$. In
this connection, the role of resonances on the asymptotic dynamics may be
relevant on a large time-scale.
\par
In this work we focus on a direct asymptotic analysis of the dynamics
generated by solutions of the Cauchy problem for the wave equation (\ref{pwe}
). The main results concern both the homogeneous case, considered in Theorem 
\ref{T1}, and the non-homogeneous case, in Theorem \ref{inhom}. In Theorem 
\ref{T-eff} a model for the effective dynamics in the small-scale limit is
provided in both cases. We next present an overview of the outcomes for the
homogeneous Cauchy problem; more detailed results on such a case as well on
the one in presence of a source term are provided in the aforementioned
theorems.
\par
Let $u_{\varepsilon }$ and $u_{\text{\textrm{free}}}$ denote the classical
solutions of the Cauchy problems (here and in the following $t\geq 0$)
\begin{equation*}
\begin{cases}
(\varepsilon ^{-2}\chi _{\Omega _{\varepsilon }}+\chi _{\mathbb{R}^{3}\backslash \Omega _{\varepsilon }})\partial _{tt}u_{\varepsilon
}(t)=\Delta u_{\varepsilon }(t)\,, \\ 
u_{\varepsilon }(0)=\phi \,, \\ 
\partial _{t}u_{\varepsilon }(0)=\psi \,,
\end{cases}%
\qquad 
\begin{cases}
\partial _{tt}u_{\text{\textrm{free}}}(t)=\Delta u_{\text{\textrm{free}}
}(t)\,, \\ 
u_{\text{\textrm{free}}}(0)=\phi \,, \\ 
\partial _{t}u_{\text{\textrm{free}}}(0)=\psi \,.
\end{cases}%
\end{equation*}%
Assume that $\phi ,\psi \in H^{6}(\mathbb{R}^{3})$ are supported outside $
\Omega _{\varepsilon }$ and define $h:[0,+\infty )\rightarrow \mathbb{R}$
depending on these Cauchy data by%
\begin{equation*}
h(t):=\int_{|y-y_{0}|<t}\frac{\Delta ^{2}\phi (y)}{4\pi |y-y_{0}|}\,dy+\frac{%
1}{4\pi t}\int_{|x-y_{0}|=t}\Delta \psi (x)\,d\sigma _{t}(y)\,.
\end{equation*}%
Notice that, by Kirchhoff's formula, $h(t)=\Delta u_{\text{\textrm{free}}%
}(t,y_{0})$. Then, for any $\tau \in (0,\frac{1}{11})$ there exists $%
\varepsilon _{0}>0$ small enough such that the expansion%
\begin{equation}
u_{\varepsilon }(t,x)=\,u_{\text{\textrm{free}}}(t,x)+\varepsilon
(\varepsilon ^{2}-1)\,H(t-|x-y_{0}|)\,\frac{q(t-|x-y_{0}|)}{4\pi |x-y_{0}|}+{%
r}_{\varepsilon }(t,x)\,,  \label{exp}
\end{equation}%
holds for all $\varepsilon \in (0,\varepsilon _{0})$ and the remainder
allows the estimate%
\begin{equation}
\sup_{0<t\leq 1/\varepsilon ^{\tau }}\Vert {r}_{\varepsilon }(t,\cdot )\Vert
_{L^{2}(\mathbb{R}^{3})}\lesssim \,o(\varepsilon )\big(\Vert \Delta ^{3}\phi
\Vert _{L^{2}(\mathbb{R}^{3})}+\Vert \Delta ^{3}\psi \Vert _{L^{2}(\mathbb{R}%
^{3})}\big)\,.  \label{rem_est}
\end{equation}%
Here, $H$ is the Heaviside function and $q(t):=\sum_{k=1}^{+\infty }q_{k}(t)$%
, with each $q_{k}(t)$ solving the Cauchy problem 
\begin{equation}
\begin{cases}
\ddot{q}_{k}(t)=-\lambda _{k}^{-1}\,\big(q_{k}(t)-|\langle e_{k},1\rangle
|^{2}\,h(t)\big)\,, \\ 
q_{k}(0)=\dot{q}_{k}(0)=0\,,%
\end{cases}
\label{td_ps}
\end{equation}%
where $e_{k}\in L^{2}(\Omega )$ is an eigenfunction of $N_{0}$ corresponding to the eigenvalue $
\lambda _{k}$.
\par
The effective dynamics given by 
\begin{equation}  \label{sph}
u_{\varepsilon,\text{\textrm{eff}}}(t,x)=\,u_{\text{\textrm{free}}%
}(t,x)+\varepsilon(\varepsilon^{2}-1) \,H(t-|x-y_{0}|)\,\frac{q(t-|x-y_{0}|)%
}{4\pi|x-y_{0}|}
\end{equation}
is the resulting point scatterer approximation of the propagation, since it
solves the Cauchy problem 
\begin{equation*}
\begin{cases}
\partial_{tt} u_{\varepsilon,\text{\textrm{eff}}}(t)=\Delta u_{\varepsilon,%
\text{\textrm{eff}}}(t)+\varepsilon(\varepsilon^{2}-1)\,q(t)\,\delta_{y_{0}} \\ 
u_{\varepsilon,\text{\textrm{eff}}}(0)=\phi \\ 
\partial_{t} u_{\varepsilon,\text{\textrm{eff}}}(0)=\psi\,,%
\end{cases}
\end{equation*}
where $\delta_{y_{0}}$ denotes the Dirac delta distribution supported at the
point $y_{0}$. The spherical wave component of $u_{\varepsilon,\text{\textrm{%
eff}}}$ has the space profile of a Green function centered at $y_{0}$
modulated by a function whose evolution in time is determined by the Cauchy
data through the solution of \eqref{td_ps}.
\par
It is worth to remark that (\ref{exp}) holds for the specific class of $%
H^{6}(\mathbb{R}^{3})$ Cauchy data with support away from $y_{0}$. This
condition is needed to control the high-energy contribution to the inverse
Laplace transform involved in the computation of the dynamics. It is unclear
if, for less regular data which may concentrate their $L^{2}$-mass close to $%
y_{0}$, a different asymptotic may occur.
\par
Since the solution $t\mapsto u_{\varepsilon}(t,\cdot)$ defines a continuous
flow in the Sobolev space $H^{2}( \mathbb{R}^{3})\subset{\mathcal{C}}_{b}(%
\mathbb{R}^{3}) $ and since \eqref{sph} is smooth outside $y_{0}$ but
neither bounded nor continuous there, the small remainder ${r}%
_{\varepsilon}(t,\cdot)$ must compensate such loss of regularity close to $%
y_{0}$. It is worth to remark that this is not in contradiction with the
estimates (\ref{rem_est}), however, such circumstance prevents to control
the errors on a more regular Sobolev scale.
\par
Although (\ref{exp})-(\ref{td_ps}) does not explicitly involve resonant
states, the resonant energies $\lambda_{k}^{-1}$ appears in the definition
of the modulation functions $q_{k}(t)$. This expansion holds on a large time
scale which is, however, shorter than the estimated resonances life-time. We
expect that (\ref{exp})-(\ref{td_ps}) can be recasted in terms of
quasi-bound states. Let us notice in turns that, representing the quantum
and the classical-waves dynamics for point perturbations of the Laplacian in
terms of resonant states presents well-known technical asperities (involving
a careful estimate of the contributions to the time-propagator from energies
close to resonances) and very few general results have been provided in this
sense, see \cite{AlbHoeg}.
\par
The small-scale expansions and the effective model presented in this work
are new. In particular, the asymptotic formulae resulting from our analysis
are given with global-in-space estimates of the errors on a large time
scale. Such formulae provide, in our opinion, a useful contribution in
applications involving time domain data. Let us recall, at this concern,
that the interaction of classical waves with \emph{micro-resonators} is
attracting the increasing interest of different branches of applied physics,
where the manipulation of the dynamics at micro-scales has several relevant
applications. For instance, these resonant-scattering phenomena have been
used for the design of acoustic and electromagnetic metamaterials, or in the
realization of contrast agents for various imaging strategies. Moreover,
they show to have a potential impact in classical and quantum information
processing. A vast literature has been devoted to these topics, see \cite%
{GaSi}, \cite{GoAl}, \cite{Kuz}, \cite{RagAl}, \cite{Stau} to cite a few.
\par
While the small-$\varepsilon $ asymptotic of the stationary scattering
problem (\ref{WE_eps_omega}) has been well understood, only partial results
have been provided so far for the corresponding dynamics. An attempt to derive a Foldy-Lax model in the time domain has been recently
discussed in the two-dimensional case using Dirichlet discs of small size, see \cite{Kach}. As regards the model equation \eqref{WE_eps_omega}, a possible approach to time-domain expansions was argued in \cite{Amm1} relying on the previous works \cite{Amm td1} and \cite[Appendix B]{Amm td}. This, however, would produce an expansion only valid after truncating the high-frequency components of the fields. To the best of our knowledge, none explicit formula has been produced within such an approach, neither the relation with the asymptotic expansion of the full time-domain solution has been discussed so far. 
In \cite{SiWa TD}, the equation \eqref{pwe}, excited by a source term but with
trivial Cauchy data, have been considered in the perspective of using
dielectric micro-resonators as contrast agents for imaging in the time
domain. The authors use known estimates in the Laplace-transform domain to
get the asymptotic expansion of the solution on a finite time interval with
a point-wise estimate of the error holding far from the scatterer location,
while, it deserves noting, none effective limit model has been discussed. In the limit of their validity, the interesting asymptotic
expansions provided in \cite{SiWa TD} are coherent with the ones provided
here, enlightening the role of the Newton eigenvalues in the construction of
the dynamics.
\begin{remark}
Without loss of generality, thanks to the translational invariance of the equations, throughout the text we set for simplicity $y_{0}=0$ in the definition (\ref{Omega_eps}) of the small domain
$\Omega_{\varepsilon}$; the statements of the Theorems \ref{T1},
\ref{inhom} and \ref{T-eff} easily adapt to the case of a generic point $y_{0}$.
\end{remark}
\section{Notation} 
\par\noindent
$\bullet$ $\Omega_{\ve}$ denotes the contracted set $\Omega_{\ve}:=\{x\in\RE^{3}: x=\ve y\,,\  y\in\Omega\}$, $0<\ve<1$, where  $\Omega\subset\RE^{3}$ is a given open, bounded and connected set containing the origin; $|\Omega|$ denotes the volume of $\Omega$. 
\vskip5pt\par\noindent
$\bullet$ $\|\cdot\|$ denotes the norm in a space of square integrable functions like $L^{2}(\Omega)$ and $L^{2}(\RE^{3})$ with scalar product $\langle\cdot,\cdot\rangle$; $\|\cdot\|$ also denotes the operator norm for bounded linear operators acting between such spaces. Norms in different spaces are specified with the appropriate subscripts. 
\vskip5pt\par\noindent
$\bullet$ $\bou(X,Y)$ denotes the Banach space of bounded operators from the Banach space $X$ to the Banach space $Y$; $\bou(X,X)\equiv\bou(X)$.
\vskip5pt\par\noindent
$\bullet$ $\mathcal L$ denotes the Laplace transform.
\vskip5pt\par\noindent
$\bullet$
$H(t)$ denotes the Heaviside function.
\vskip5pt\par\noindent
$\bullet$ $\varrho(A)$ and $\sigma(A)$ denote the resolvent set and the spectrum of the self-adjoint operator $A$; $\sigma_{d}(A)$ denotes its discrete spectrum.
\vskip5pt\par\noindent
$\bullet$ $H^{2k}(\RE^{3})$, $k\in {\mathbb N}$, denotes the Hilbert-Sobolev space of order $2k$, i.e. the space of square integrale functions $f:\RE^{3}\to\CO$ such that $\Delta^{k} f$ is square integrable, $\Delta^{k}$ denoting the $k$-th power of the Laplacian.  
\vskip5pt\par\noindent
$\bullet$ $\mathcal G_{z}$, $z\in\CO\backslash(-\infty,0]$, denotes the kernel function of $(-\Delta+z^{2})^{-1}$, i.e., $\mathcal G_{z}(x)=\frac{e^{-z|x|}}{4\pi|x|}$.
\vskip5pt\par\noindent
$\bullet$ $ f(t,\ve)\lesssim g(t,\ve)$ means that there exists $\kappa\ge 0$, independent of $t$ and $\ve$ such that $f(t,\ve)\le \kappa\, g(t,\ve)$.
\section{The model operator and the corresponding wave equation}
The wave equation \eqref{pwe} re-writes as 
$$
\partial_{tt}u=A(\ve)u\,,
$$
where the linear operator in $L^{2}(\RE^{3})$ is defined by   
$$A(\ve):H^{2}(\RE^{3})\subset L^{2}(\RE^{3})\to L^{2}(\RE^{3})\,,\qquad 
A(\ve):=(\ve^{2}\chi_{\Omega_{\ve}}+\chi_{{\RE^{3}\backslash\Omega_{\ve}}})\Delta\,.
$$ 
It corresponds to an additive perturbation
of the Laplacian:
\begin{equation}%
\begin{array}
[c]{ccc}%
A(\ve)=\Delta-T(\varepsilon)\,, &  & T(\varepsilon):=(  
1-\varepsilon^{2})    \chi_{\Omega_{\varepsilon}}\Delta\,.
\end{array}
\label{H_eps_add}%
\end{equation}
Let us set
$$
R_{z}:=(-\Delta+z)^{-1}\,,\qquad z\in\CO\backslash(-\infty,0]\,.
$$
We introduce the Hilbert space 
$$
L_{\ve}^{2}(\RE^{3}):=L^{2}(\RE^{3};(\ve^{-2}\chi_{\Omega_{\ve}}+\chi_{{\RE^{3}\backslash\Omega_{\ve}}})dx)\,;
$$
the corresponding scalar product and norm are denoted by $\langle\cdot,\cdot\rangle_{\ve}$ and 
$\|\cdot\|_{\ve}$ respectively; By the inequalities
\be\label{equiN}
\|u\|\le \|u\|_{\ve}\le\ve^{-1}\|u\|\,,
\ee
one gets the equivalence, holding in the Banach space sense,
\be\label{equi}
L_{\ve}^{2}(\RE^{3})\simeq L^{2}(\RE^{3})\,.
\ee
\begin{remark}\label{norm} Notice that, if $u\in L^{2}(\RE^{3})$ is supported outside $\Omega_{\ve}$, then $\|u\|_{\ve}=\|u\|$. Therefore, using  $\|\cdot\|_{\ve}$ to also denote the norm in $\bou(L_{\ve}^{2}(\RE^{3}))$, one has, for any $L\in\bou(L_{\ve}^{2}(\RE^{3}))$  and for any $u\in L^{2}(\RE^{3})$ supported outside $\Omega_{\ve}$, 
$$
\|Lu\|\le\|L\|_{\ve}\|u\|\,.
$$ 
\end{remark}
\begin{theorem}\label{TK} $A(\ve)$ is a closed operator in $L^{2}(\mathbb{R}^{3})$ and a non-positive self-adjoint operator in $L_{\ve}^{2}(\RE^{3})$; it has resolvent  
\be\label{res}
R_{z}(\ve):=(-A(\ve)+z)^{-1}=R_{z}-K_{z}(\ve)\,,\qquad z\in \CO\backslash(-\infty,0]\,,
\ee
where
$$
K_{z}(\ve):=R_{z}(1+T(\varepsilon)R_{z})^{-1}T(\varepsilon)R_{z}\,,\qquad z\in \CO\backslash(-\infty,0]\,.
$$
\end{theorem}
\begin{proof} Since both $(\ve^{2}\chi_{\Omega_{\ve}}+\chi_{{\RE^{3}\backslash\Omega_{\ve}}})$ and $(\ve^{-2}\chi_{\Omega_{\ve}}+\chi_{{\RE^{3}\backslash\Omega_{\ve}}})$ are bounded and $\Delta:H^{2}(\RE^{3})\subset L^{2}(\RE^{3})\to L^{2}(\RE^{3})$ is closed, $A(\ve)=(\ve^{2}\chi_{\Omega_{\ve}}+\chi_{{\RE^{3}\backslash\Omega_{\ve}}})\Delta$ is closed as well; furthermore, by $$\langle u, A(\ve)u\rangle_{\ve}=\langle u, \Delta u\rangle\le 0\,,$$ 
one gets $A(\ve)\le 0$ and $\sigma(A(\ve))\subseteq (\infty,0]$. \par
Let $z\in \CO\backslash(-\infty,0]$ such that $\text{Re}(z)\ge0$; the identity%
\begin{align*}
\left\Vert u\right\Vert  ^{2}=\,&\left\Vert (  -\Delta+z)
R_{z}u\right\Vert  ^{2}
=
\left\Vert -\Delta R_{z}%
u\right\Vert  ^{2}+|z|^{2}\left\Vert
R_{z}u\right\Vert  ^{2}+2\text{Re}(z)\langle-\Delta R_{z}u,R_{z}u\rangle
\end{align*}
implies%
\[
\left\Vert -\Delta R_{z}\right\Vert  \leq1\,.
\]
By \eqref{H_eps_add}, this gives 
\be\label{TR}
\left\Vert T(\varepsilon)R_{z}\right\Vert  \leq(
1-\varepsilon^{2})  \left\Vert \Delta R_{z}\right\Vert 
\leq(  1-\varepsilon^{2})<1\,.
\ee
Hence, 
\be\label{ns}
\sum_{n=0}^{+\infty}(-1)^{n}  (T(\varepsilon)
R_{z})  ^{n}=(  1+T(\varepsilon)R_{z})  ^{-1}\,,
\ee
where the Neumann series on the left converges in  ${\bou}(  L^{2}(  \mathbb{R}^{3}))  $ for any $z\in \CO\backslash(-\infty,0]$ such that $\text{Re}(z)\ge0$.  
Furthermore, for such a $z$, by simple algebraic manipulations, one gets
\be\label{mani}
R_{z}(  1+T(\varepsilon)R_{z})  ^{-1}
=R_{z}-R_{z}(1+T(\varepsilon)R_{z})^{-1}T(\varepsilon)R_{z}\,,
\ee
and
$$
(-A(\ve)+z)R_{z}(  1+T(\varepsilon)R_{z})  ^{-1}={\mathbb 1}_{L^{2}(\RE^{3})}\,,
$$
$$R_{z}(  1+T(\varepsilon)R_{z})  ^{-1}(-A(\ve)+z)={\mathbb 1}_{H^{2}(\RE^{3})}\,.
$$
By the equivalence \eqref{equi}, $A(\ve)$ is closed as an operator in $L_{\ve}^{2}(\RE^{3})$; furthermore, $A(\ve)$ is symmetric in $L_{\ve}^{2}(\RE^{3})$. Hence,  by $$\ran(-A(\ve)\pm i)=\dom(R_{z}(  1+T(\varepsilon)R_{\pm i})  ^{-1})=L_{\ve}^{2}(\RE^{3})\,,$$  $A(\ve)$ is self-adjoint as an operator in $L_{\ve}^{2}(\RE^{3})$. \par 
Finally, by \cite[Theorem 2.19 and Remark 2.20]{CFP}, the resolvent formula \eqref{mani} extends to $\varrho(A(\ve))\cap\varrho(\Delta)=\CO\backslash(-\infty,0]$.
\end{proof}
\begin{remark} One can show that $\sigma(A(\ve))=\sigma_{ac}(A(\ve))=(-\infty,0]$. However, for the purposes of this work, it suffices to know, by $A(\ve)\le 0$, that $\sigma(A(\ve))\subseteq(-\infty,0]$.  
\end{remark}
\begin{corollary}\label{coro} For any $z\in\CO\backslash(-\infty,0]$, for any integer $n\ge 1$ and for any $u\in H^{2n}(\RE^{3})$ supported outside $\Omega_{\ve}$, one has
$$
K_{z}(\ve)u=\frac1{z^{n}}\,K_{z}(\ve)\Delta^{n}u\,.
$$
\end{corollary}
\begin{proof} For any self-adjoint operator $A$ and for any $z\in\varrho(A)\cap\CO\backslash\{0\}$, one has 
$$
(-A+z)^{-1}=\frac1{z}+\frac1{z}\,(-A+z)^{-1}A\,.
$$
Iterating, one gets
\be\label{zn}
(-A+z)^{-1}=\sum_{k=1}^{n}\frac1{z^{k}}\,A^{k-1}+\frac1{z^{n}}\,(-A+z)^{-1}A^{n}\,.
\ee
By $A(\ve)u=\Delta u$ for any $u$ supported outside $\Omega_{\ve}$, by \eqref{res} and by applying \eqref{zn} to both $A=\Delta$ and $A=A(\ve)$,  one gets, whenever $u\in H^{2n}(\RE^{3})$, 
\begin{align*}
R_{z}(\ve)u=&\,\sum_{k=1}^{n}\frac1{z^{k}}\,A(\ve)^{k-1}u+\frac1{z^{n}}\,R_{z}(\ve)A(\ve)^{n}u
\\
=&\,\sum_{k=1}^{n}\frac1{z^{k}}\,\Delta^{k-1}u+\frac1{z^{n}}\,R_{z}\Delta^{n}u-\frac1{z^{n}}\,K_{z}(\ve)\Delta^{n}u\\
=&\,R_{z}u-\frac1{z^{n}}\,K_{z}(\ve)\Delta^{n}u\\
=&\,R_{z}(\ve)u+K_{z}(\ve)u-\frac1{z^{n}}\,K_{z}(\ve)\Delta^{n}u\,.
\end{align*}
\end{proof}
Given a non-positive self-adjoint operator $A$ in a Hilbert space ${\mathcal H}$, we consider the Cauchy problem for the corresponding wave equation, i.e., 
\be
\begin{cases}\label{Cpb}
\partial_{tt} u(t)=Au(t)\\
u(0)=\phi\\
\partial_{t} u(0)=\psi\,,
\end{cases}
\ee
where both the data $\phi$ and $\psi$ are in ${\mathcal H}$. We say that $u\in {\mathcal C}(\RE_{+};{\mathcal H})$ is a mild solution of \eqref{Cpb} whenever
$$
\int_{0}^{t}\int_{0}^{s}u(r)\,dr\,ds\equiv\int_{0}^{t}(t-s)u(s)\,ds\,\in \D(A)
$$
and
$$
u(t)=\phi+t\psi+A\int_{0}^{t}(t-s)u(s)\,ds
$$
for any $t\ge 0$. 
By \cite[Proposition 3.14.4, Corollary 3.14.8 and Example 3.14.16]{Arendt}, the unique mild solution of \eqref{Cpb} is given by 
\be\label{sol}
u(t)=\cos(t(-A)^{1/2})\, \phi+(-A)^{-1/2}\sin(t(-A)^{1/2})\,\psi\,,
\ee
where the sine and cosine operator  functions are defined through the functional calculus for the self-adjoint $A$. By \cite[Theorem 7.4 in Chapter 2, Section 8]{Gold}, if $\phi\in \dom(A)$ and $\psi\in\dom((-A)^{1/2})$, then $u$ in \eqref{sol} is a classical solution, i.e., 
$$
u\in {\mathcal C}^{2}(\RE_{+},{\mathcal H})\cap{\mathcal C}^{1}(\RE_{+},\dom((-A)^{1/2})\cap{\mathcal C}(\RE_{+},\dom(A))\,.
$$
The Laplace transforms of the sine and the cosine functions give the  relations  
$$
\int_{0}^{\infty}e^{-\lambda\, t}\sin(tx)\,dt=x(x^{2}+\lambda^{2})^{-1}\,,
\qquad\lambda>0\,,
$$
$$
\int_{0}^{\infty}e^{-\lambda\, t}
\cos(tx))\,dt=\lambda(x^{2}+\lambda^{2})^{-1}\,,\qquad \lambda>0\,.
$$
Thus, 
by functional calculus, by the inversion of the Laplace-Stieltjes transform of Lipschitz, $\bou(L^{2}(\RE^{n}))$-valued, functions $F$ such that $F(0)=0$ (see \cite[Theorem 2.3.4]{Arendt}), one gets
\begin{align*}
(-A)^{-1/2}\sin(t(-A)^{1/2})\psi&\,=\lim_{\gamma\nearrow\infty}\frac1{2\pi i}\int_{c-i\gamma}
^{c+i\gamma} e^{ tz} (-A+z^{2})^{-1}\psi\,dz\,,\qquad\psi\in L^{2}(\RE^{3})
\end{align*}
and
\begin{align*}
(\cos(t(-A)^{1/2})-1)\phi&\,=\lim_{\gamma\nearrow\infty}\frac1{2\pi i}\int_{c'-i\gamma}
^{c'+i\gamma} e^{ tz} (-A+z^{2})^{-1}A\phi\ \frac{dz}{z}\,,\qquad \phi\in \D(A)\,,
\end{align*}
where $c>0$ and $c'>0$ are arbitrary. \par
By applying these results to both $A=A(\ve)$ and $A=\Delta$ one gets the following \begin{lemma}\label{cp-ep} Let $\phi\in H^{2(m+1)}(\RE^{3})$ and $\psi\in H^{2n}(\RE^{3})$, $m\ge 0$ and $n>0$, both having support disjoint from $\Omega_{\ve}$; let 
$u_{\ve}$ and $u_{\text {\rm free}}$ be the classical solutions of the Cauchy problems 
$$
\begin{cases}
(\ve^{-2}\chi_{\Omega_{\ve}}+\chi_{{\RE^{3}\backslash\Omega_{\ve}}})\partial_{tt} u_{\ve}(t)=\Delta u_{\ve}(t)\\
u_{\ve}(0)=\phi\\
\partial_{t} u_{\ve}(0)=\psi\,.
\end{cases}
\qquad
\begin{cases}
\partial_{tt} u_{\text {\rm free}}(t)=\Delta u_{\text {\rm free}}(t)\\
u_{\text {\rm free}}(0)=\phi\\
\partial_{t}u_{\text {\rm free}}(0)=\psi\,.
\end{cases}
$$
Then, 
$$
u_{\ve}(t)=u_{\text {\rm free}}(t)-\frac1{2\pi i}\int_{c-i\infty}
^{c+i\infty} \!\!e^{ tz} K_{z^{2}}(\ve)\Delta^{n}\psi\ \frac{dz}{z^{2n}}-\frac1{2\pi i}\int_{c'-i\infty}
^{c'+i\infty}\!\! e^{ tz} K_{z^{2}}(\ve)\Delta^{m+1}\phi\ \frac{dz}{z^{2m+1}}\,,
$$
where $c>0$ and $c'>0$ are arbitrary. 
\end{lemma}
\begin{proof}
One gets, by \eqref{res} and by Corollary \ref{coro}, for any $\phi\in H^{2(m+1)}(\RE^{3})$ and $\psi\in H^{2n}(\RE^{3})$ both having support disjoint from $\Omega_{\ve}$,
\begin{align*}
&(-A(\ve))^{-1/2}\sin(t(-A(\ve))^{1/2})\psi\\
=&\,\lim_{\gamma\nearrow\infty}\frac1{2\pi i}\int_{c-i\gamma}
^{c+i\gamma} e^{ tz} R_{z^{2}}\psi\,dz-\lim_{\gamma\nearrow\infty}\frac1{2\pi i}\int_{c-i\gamma}
^{c+i\gamma}e^{ tz} K_{z^{2}}(\ve)\Delta^{n}\psi\ \frac{dz}{z^{2n}}\\
=&\,(-\Delta)^{-1/2}\sin(t(-\Delta)^{1/2})\psi-\lim_{\gamma\nearrow\infty}\frac1{2\pi i}\int_{c-i\gamma}
^{c+i\gamma}e^{ tz} K_{z^{2}}(\ve)\Delta^{n}\psi\ \frac{dz}{z^{2n}}
\end{align*}
and
\begin{align*}
&(\cos(t(-A(\ve))^{1/2})-1)\phi\\
=&\,\lim_{\gamma\nearrow\infty}\frac1{2\pi i}\int_{c-i\gamma}
^{c+i\gamma} e^{ tz} R_{z^{2}}\Delta\phi\ \frac{dz}{z}-\lim_{\gamma\nearrow\infty}\frac1{2\pi i}\int_{c'-i\gamma}
^{c'+i\gamma} e^{ tz} K_{z^{2}}(\ve)\Delta^{m+1}\phi\ \frac{dz}{z^{2m+1}}\\
=&\,(\cos(t(-\Delta)^{1/2})-1)\phi-\lim_{\gamma\nearrow\infty}\frac1{2\pi i}\int_{c'-i\gamma}
^{c'+i\gamma} e^{ tz} K_{z^{2}}(\ve)\Delta^{m+1}\phi\ \frac{dz}{z^{2m+1}}\,.
\end{align*}
Writing $z=c+i\gamma$, $c>0$, $\gamma\in\RE$, one has
$$
\|R_{z^{2}}\|\le \begin{cases}{|z|^{-2}} &|\gamma|<c\\
{(2c|\gamma|)^{-1}}&|\gamma|\ge c\end{cases}
\,,\quad \|R_{z^{2}}(\ve)\|_{\ve}\le\begin{cases}{|z|^{-2}} &|\gamma|<c\\
{(2c|\gamma|)^{-1}}&|\gamma|\ge c\end{cases}\,,
$$
and so, by Remark \ref{norm}, whenever $\varphi\in L^{2}(\RE^{3})$ is supported outside $\Omega_{\ve}$,
\be\label{KE}
\|K_{z^{2}}(\ve)\varphi\|\le 
\|R_{z^{2}}\varphi\|+\|R_{z^{2}}(\ve)\varphi\|\le \|\varphi\|\begin{cases}{2|z|^{-2}} &|\gamma|<c\\
{(c|\gamma|)^{-1}}&|\gamma|\ge c\,.\end{cases} 
\ee
Hence, the above improper operator-valued integrals are absolutely convergent and one gets,  
$$
(-A(\ve))^{-1/2}\sin(t(-A(\ve))^{1/2})\psi=(-\Delta)^{-1/2}\sin(t(-\Delta)^{1/2})\psi-\frac1{2\pi i}\int_{c-i\infty}
^{c+i\infty} \!\!e^{ tz} K_{z^{2}}(\ve)\Delta^{n}\psi\ \frac{dz}{z^{2n}}
$$
and
$$
\cos(t(-A(\ve))^{1/2})\phi=\cos(t(-\Delta)^{1/2})\phi-\frac1{2\pi i}\int_{c'-i\infty}
^{c'+i\infty}\!\! e^{ tz} K_{z^{2}}(\ve)\Delta^{m+1}\phi\ \frac{dz}{z^{2m+1}}\,,
$$
where $c>0$ and $c'>0$ are arbitrary. The proof is then concluded by \eqref{sol}.
\end{proof}
\begin{remark} By $|\Omega_{\ve}|\searrow 0$ as $\ve\searrow 0$, for any $u\in H^{2}(\RE^{3})$ one has
$$
\lim_{\ve\searrow 0}\|A(\ve)u-\Delta u\|^{2}=\lim_{\ve\searrow 0}\,(1-\ve^{2})^{2}\!\int_{\Omega_{\ve}}|\Delta u(x)|^{2}\,dx=0\,.
$$
Therefore, by \cite[Theorems 2.25 and 2.29]{Kato}, $A(\ve)$ converges to the free Laplacian in norm resolvent sense, i.e., for any $z\in\CO\backslash(-\infty,0]$, there holds
\be\label{norm-conv}
\lim_{\ve\searrow 0}\left\|(-A(\ve)+z)^{-1}-(-\Delta +z)^{-1}\right\|=0\,.
\ee
If the sine and the cosine operator function of $A(\ve)$ had norms in $\B(L^{2}(\RE^{3}))$ uniformly bounded with respect to $\ve$, then \eqref{norm-conv} would imply the convergence of the solutions of the Cauchy problem for the wave equation for $A(\ve)$ to the ones for $\Delta$, and that would be true for generic initial data (see \cite[Theorem 8.6 in Section 8, Chapter 2]{Gold}). However, by \eqref{equiN}, these operator norms behave like $1/\ve$, and so, the convergence of the dynamics generated by $A(\ve)$ to the free one is not guaranteed. As we will prove in the following sections, thanks to Lemma  \ref{cp-ep}, convergence to the free case surely holds whenever the initial data are supported outside $\Omega_{\ve}$.
\end{remark}
\section{Operator estimates}
In the following, we use the
identification%
\[
L^{2}(  \mathbb{R}^{3})  \equiv L^{2}(  \RE^{3}\backslash\Omega_{\ve})    \oplus L^{2}(    \Omega_{\varepsilon})    \,,
\]
provided by the unitary map $u\mapsto    1_{\RE^{3}\backslash\Omega_{\ve}}u\oplus 1_{\Omega_{\ve}}$.
Here and below, given the measurable domain $D\subset\RE^{3}$ we denote by $1_{D}: L^{2}(\RE^{3})\to L^{2}(D)$ the bounded linear operator given by the restriction to $D$; then, the adjoint $1^{*}_{D}:L^{2}(D)\to L^{2}(\RE^{3})$ corresponds to the extension by zero. Notice that $1^{*}_{D}1_{D}u=\chi_{D}u$.\par
In such a framework, the free resolvent rewrites as the  operator block matrix 
\begin{equation}
R_{z}=%
\begin{bmatrix}
1_{{\RE^{3}\backslash\Omega_{\ve}}}R_{z}1^{*}_{{\RE^{3}\backslash\Omega_{\ve}}} & 1_{\Omega
_{\varepsilon}^{c}}R_{z}1^{*}_{\Omega_{\varepsilon}}\\
1_{\Omega_{\varepsilon}}R_{z}1^{*}_{{\RE^{3}\backslash\Omega_{\ve}}} & 1_{\Omega
_{\varepsilon}}R_{z}1^{*}_{\Omega_{\varepsilon}}%
\end{bmatrix}:L^{2}(  \RE^{3}\backslash\Omega_{\ve})    \oplus L^{2}(    \Omega_{\varepsilon})\to L^{2}(    \RE^{3}\backslash\Omega_{\ve})    \oplus L^{2}(    \Omega_{\varepsilon})
\,. \label{Res_matrix}%
\end{equation}
By the identity 
$$
T(\varepsilon)R_{z}=(    1-\varepsilon^{2})    \chi_{\Omega_{\varepsilon
}}\Delta R_{z}=-(    1-\varepsilon^{2})    1_{\Omega_{\varepsilon}%
}(    1-zR_{z})\,,    
$$
making use of \eqref{Res_matrix}, we get%
\begin{align*}
1+T(\varepsilon)R_{z}&=\begin{bmatrix}
1 &0 \\
0& 1%
\end{bmatrix}
-(    1-\varepsilon^{2})    \left(  
\begin{bmatrix}
0& 0\\
0& 1%
\end{bmatrix}
-z%
\begin{bmatrix}
0& 0\\
1_{\Omega_{\varepsilon}}R_{z}1^{*}_{{\RE^{3}\backslash\Omega_{\ve}}} & 1_{\Omega
_{\varepsilon}}R_{z}1^{*}_{\Omega_{\varepsilon}}%
\end{bmatrix}
\right)    \\
& \\
&  \left.  =%
\begin{bmatrix}
1 & 0\\
(    1-\varepsilon^{2})  z  1_{\Omega_{\varepsilon}}R_{z}%
1^{*}_{{\RE^{3}\backslash\Omega_{\ve}}} &     \varepsilon
^{2}+(    1-\varepsilon^{2})    z1_{\Omega_{\varepsilon}}R_{z}    1^{*}_{\Omega_{\varepsilon}}%
\end{bmatrix}
\,.\right.
\end{align*}
Hence, for any $z\in\CO$ such that $z^{2}\in\CO\backslash(-\infty,0]$, introducing the notations
$$
S_{z}(    \varepsilon)  :=    \varepsilon^{2}+(    1-\varepsilon^{2})  
z^{2}1_{\Omega
_{\varepsilon}}R_{z^{2}}   1^{*}_{\Omega_{\varepsilon}}\,,
$$
\be
P_{z}(    \varepsilon):=(  
1-\varepsilon^{2})    S_{z}(    \varepsilon)^{-1}  
z^{2}1_{\Omega_{\varepsilon}}R_{z^{2}}1^{*}_{\Omega_{\varepsilon
}^{c}}\,, \label{K_k_def}%
\ee
we get%
\begin{equation}
( 1+T(\varepsilon)R_{z^{2}})    ^{-1}=%
\begin{bmatrix}
1 & 0\\
P_{z}(    \varepsilon)    & S_{z}(    \varepsilon)^{-1}  
\end{bmatrix}
\label{KR_id_1}%
\end{equation}
and so
\begin{align*}
(1+T(\varepsilon)R_{z^2})^{-1}T(\varepsilon)=\,&(1-\ve^{2})
\begin{bmatrix}
1 & 0\\
P_{z}(    \varepsilon)    & S_{z}(    \varepsilon)^{-1}  
\end{bmatrix}\begin{bmatrix}
0& 0\\
1_{\Omega_{\varepsilon}}\Delta1^{*}_{{\RE^{3}\backslash\Omega_{\ve}}} & 1_{\Omega
_{\varepsilon}}\Delta1^{*}_{\Omega_{\varepsilon}}
\end{bmatrix}\\
=\,&(1-\ve^{2})\begin{bmatrix}
0& 0\\
S_{z}(    \varepsilon)^{-1}1_{\Omega_{\varepsilon}}\Delta1^{*}_{{\RE^{3}\backslash\Omega_{\ve}}} & S_{z}(    \varepsilon)^{-1}1_{\Omega
_{\varepsilon}}\Delta1^{*}_{\Omega_{\varepsilon}}
\end{bmatrix}\\
=\,&(1-\ve^{2})1^{*}_{\Omega_{\varepsilon}}S_{z}(    \varepsilon)^{-1}1_{\Omega_{\varepsilon}}\Delta\,.
\end{align*}
This gives
\begin{align*}
K_{z^{2}}(\ve)=R_{z^{2}}(1+T(\varepsilon)R_{z^2})^{-1}T(\varepsilon)R_{z^{2}}=\,&(1-\ve^{2})R_{z^{2}}1^{*}_{\Omega_{\varepsilon}}S_{z}(    \varepsilon)^{-1}1_{\Omega_{\varepsilon}}\Delta{R}_{z^{2}}
\,.
\end{align*}
Introducing the unitary dilation operator
$$
U(\ve):L^{2}(\RE^{3})\to L^{2}(\RE^{3})\,,\qquad U(\ve)u(x):=\ve^{3/2}u(\ve x)\,,
$$
one has
$$
U(\ve)\chi_{\Omega_{\ve}}=\chi_{\Omega}U(\ve)\,,\qquad 
U(\ve)\chi_{{\RE^3\backslash\Omega_{\ve}}}=\chi_{{\RE^3\backslash\Omega}}U(\ve)
$$
and 
$$
U(\ve)R_{z^{2}}=\ve^{2}R_{(\ve z)^{2}}U(\ve)\,.
$$
Hence, defining
\begin{equation}
N_{z}:=1_{\Omega}R_{z^{2}}1^{*}_{\Omega}:L^{2}(\Omega)\to L^{2}(\Omega)\,,
\label{N_k_def}%
\end{equation}
one has
\begin{align*}
&1_{\Omega_{\ve}}R_{z^{2}}1^{*}_{\Omega_{\ve}}=\ve^{2}1_{\Omega_{\ve}}U(\ve)^{*}R_{(\ve z)^{2}}U(\ve)1^{*}_{\Omega_{\ve}}\\
=\,&\ve^{2}1_{\Omega_{\ve}}U(\ve)^{*}(\chi_{\Omega}+\chi_{{\RE^3\backslash\Omega}})R_{(\ve z)^{2}}(\chi_{\Omega}+\chi_{{\RE^3\backslash\Omega}})U(\ve)1^{*}_{\Omega_{\ve}}\\
=\,&\ve^{2}1_{\Omega_{\ve}}(U(\ve)^{*}\chi_{\Omega}+\chi_{{\RE^{3}\backslash\Omega}}U(\ve)^{*})R_{(\ve z)^{2}}(\chi_{\Omega}U(\ve)+U(\ve)\chi_{{\RE^{3}\backslash\Omega}})1^{*}_{\Omega_{\ve}}\\
=\,&\ve^{2}1_{\Omega_{\ve}}U(\ve)^{*}\chi_{\Omega}R_{(\ve z)^{2}}\chi_{\Omega}U(\ve)1^{*}_{\Omega_{\ve}}\\
=\,&\ve^{2}1_{\Omega_{\ve}}U(\ve)^{*}1^{*}_{\Omega}1_{\Omega}R_{(\ve z)^{2}}1^{*}_{\Omega}1_{\Omega}U(\ve)1^{*}_{\Omega_{\ve}}\\
=\,&\ve^{2}1_{\Omega_{\ve}}U(\ve)^{*}1^{*}_{\Omega}N_{\ve z}1_{\Omega}U(\ve)1^{*}_{\Omega_{\ve}}\,.
\end{align*}
Defining
$$
{M}_{z}(\ve):=   1+(    1-\varepsilon^{2})  
z^{2}{N}_{\varepsilon z}\,,
$$
there follows 
\begin{equation}
S_{z}(    \varepsilon)    =\varepsilon^{2}%
1_{\Omega_{\ve}}U^{*}_{\varepsilon}1^{*}_{\Omega}{M}_{z}(\ve)1_{\Omega}U_{\varepsilon}1^{*}_{\Omega_{\ve}}\qquad \label{S_k_id}%
\end{equation}
and then, defining 
\be\label{gze}
G_{z}(\ve):=\ve^{-3/2}R_{z^{2}}\chi_{\Omega_{\ve}}U_{\varepsilon}^{*}1^{*}_{\Omega}\equiv
\ve^{-3/2}R_{z^{2}}U_{\varepsilon}^{*}1^{*}_{\Omega}
\ee
one finally gets the following
\begin{lemma} Let $K_{z}(\ve)$ be as in Theorem \ref{TK}. Then, for any $u\in H^{2}(\RE^{3})$ and for any $z\in\CO$ such that $z^{2}\in\CO\backslash(-\infty,0]$,
\be\label{Kz}
K_{z^{2}}(\ve)u=
(1-\ve^{2})\ve\, G_{z}(\ve)
{M}_{z}(\ve)^{-1}G_{\bar z}(\ve)^{*}\Delta u\,.
\ee
\end{lemma}
\begin{proof}
\begin{align*}
K_{z^{2}}(\ve)u=\,&R_{z^{2}}(1+T(\varepsilon)R_{z^2})^{-1}T(\varepsilon)R_{z^{2}}u\\
=\,&(1-\ve^{2})\ve^{-2}
R_{z^{2}}1^{*}_{\Omega_{\varepsilon}}1_{\Omega_{\varepsilon}}
U_{\varepsilon}^{*}1_{\Omega}{M}_{z}(\ve)^{-1}1^{*}_{\Omega}U_{\varepsilon}1_{\Omega_{\ve}}^{*}1_{\Omega_{\varepsilon}}\Delta{R}_{z^{2}}u
\\
=\,&(1-\ve^{2})\ve^{-2}
R_{z^{2}}\chi_{\Omega_{\ve}}U_{\varepsilon}^{*}1^{*}_{\Omega}
{M}_{z}(\ve)^{-1}1_{\Omega}U_{\varepsilon}\chi_{\Omega_{\ve}}{R}_{z^{2}}\Delta u
\\
=\,&
(1-\ve^{2})\ve G_{z}(\ve)
{M}_{z}(\ve)^{-1}G_{\bar z}(\ve)^{*}\Delta u
\,.
\end{align*}
\end{proof}
\subsection{The Newton potential operator of $\Omega$} Now, we introduce the not negative, symmetric  operator $N_{0}$ in $L^{2}(\Omega)$ defined by 
$$
N_{0}:L^{2}(\Omega)\to L^{2}(\Omega)\,,\qquad 
N_{0}u(x):=\frac1{4\pi}\int_{\Omega}\frac{u(y)\,dy}{|x-y|}\,.
$$
By Sobolev's inequality (see, e.g., \cite[Appendix 2 to I.1]{Sim})
$$
\int_{\Omega\times\Omega}\frac{dx\,dy}{|x-y|^{2}}=
\int_{\RE^{3}\times\RE^{3}}\chi_{\Omega}(x)\chi_{\Omega}(y)\,\frac{dx\,dy}{|x-y|^{2}}\lesssim \|\chi_{\Omega}\|^{2}_{L^{3/2}(\RE^{3})}
$$
and so $N_{0}$ is Hilbert-Schmidt and hence compact (see, e.g., \cite[Theorem A.28]{Sim}). 
Since $-\Delta N_{0}u=u$, one gets $\ker(N_{0})=\{0\}$. 
By the spectral theory of compact symmetric operators (see, e.g., \cite[Section 6]{Jo}), $\sigma_{d}(N_{0})=\sigma(N_{0})\backslash\{0\}$ and the orthonormal 
sequence $\{e_{k}\}_{1}^{+\infty}$ of eigenvectors corresponding to the discrete spectrum is an orthonormal base of $\ker(N_{0})^{\perp}=L^{2}(\Omega)$. We denote by $\{\lambda_{k}\}_{1}^{+\infty}$, the set of eigenvalues (counting multiplicities) indexed in decreasing order. One has $\lambda_{1}=\|N_{0}\|$ and $\lambda_{k}\searrow 0$. Furthermore, for any $u\in L^{2}(\Omega)$ and for any bounded measurable function $f:[\,0,\|N_{0}\|\,]\to\CO$,
\be\label{sp-res}
f(N_{0})u=\sum_{k=1}^{+\infty}f(\lambda_{k})\langle e_{k},u\rangle e_{k}\,.
\ee
\begin{lemma}\label{LN}  Let $z=c+i\gamma$, $c>0$, $\gamma\in\RE$. Then,  $1+z^{2}N_{0}$ has a bounded inverse and
$$
\| (1+z^{2}N_{0})^{-1}\|\le 
\begin{cases}
1&|\gamma|\le c\\
{|z|^{2}}{(2c|\gamma|\,)^{-1}}&|\gamma|> c
\end{cases}\quad.
$$
\end{lemma}
\begin{proof} By $\bar z^{2}=(c^{2}-\gamma^{2})-i2c\gamma\in\CO\backslash (-\infty,0]$ and by $\sigma(-N_{0})\subset \left[-\|N_{0}\|,0\,\right]$, one gets $z^{-2}=|z|^{-4}\bar z^{2}\in \varrho(-N_{0})$. Hence,
$$
(1+z^{2}N_{0})={z^{2}}\,(N_{0}+z^{-2})
$$
has a bounded inverse and, by the functional calculus for $N_{0}$, 
$$
\|(1+z^{2}N_{0})^{-1}\|=\sup_{\lambda\in\sigma(N_{0})}\,\frac1{|1+z^{2}\lambda|}\le 
\sup_{0\le\lambda\le\|N_{0}\|}\,\frac1{|1+z^{2}\lambda|}=\left(\inf_{0\le\lambda\le\|N_{0}\|}\,|1+z^{2}\lambda|\right)^{-1}\,.
$$
The quadratic polynomial
$$
p(\lambda):=|1+z^{2}\lambda|^{2}=\big(1+(c^{2}-\gamma^{2})\lambda\big)^{2}+4c^{2}\gamma^{2}\lambda^{2}=
(\gamma^{2}+c^{2})^{2}\lambda^{2}-2(\gamma^{2}-c^{2})\lambda+1
$$
has minimum at $\lambda=\lambda_{\circ}$,
$$
\lambda_{\circ}:=\frac{\gamma^{2}-c^{2}}{(\gamma^{2}+c^{2})^{2}}\,,
$$ 
where it assumes the value 
$$
p(\lambda_{\circ})=1-\frac{(\gamma^{2}-c^{2})^{2}}{(\gamma^{2}+c^{2})^{2}}=\frac{4c^{2}\gamma^{2}}{|z|^{4}}\,.
$$ 
Therefore, 
$$
\inf_{0\le\lambda\le\|N_{0}\|}\,p(\lambda)=
\begin{cases}p(0)&\lambda_{\circ}< 0\\
p(\lambda_{\circ})&0\le\lambda_{\circ}\le\|N_{0}\|\\
p(\|N_{0}\|)&\lambda_{\circ}>\|N_{0}\|\,\,.
\end{cases}
$$
However, since $\lambda_{\circ}\searrow 0$ as $|\gamma|\nearrow+\infty$, for our later purposes it suffices to use the  rougher estimate 
$$
\inf_{0\le\lambda\le\|N_{0}\|}\,p(\lambda)\ge\inf_{\lambda\ge 0}\,p(\lambda)=
\begin{cases}p(0)&\lambda_{\circ}< 0\\
p(\lambda_{\circ})&\lambda_{\circ}\ge 0\end{cases}
=\begin{cases}1&|\gamma|\le c\\
4c^{2}\gamma^{2}|z|^{-4}&|\gamma|> c\ .
\end{cases}
$$
\end{proof}
\begin{lemma}\label{LM} Let $z\in\CO\backslash(-\infty,0]$. Then
$$
{M}_{z}(\ve)=(1+z^{2}N_{0})+{M}^{(1)}_{z}(\ve)\,,
$$
where
$$
\|{M}^{(1)}_{z}(\ve)\|\le \ve\,|z|^{2}\left(
{\frac{|\Omega|}{4\pi}}^{1/2}|z|+\|N_{0}\|\right)\,.
$$
\end{lemma}
\begin{proof} One has
\begin{align*}
{M}^{(1)}_{z}(\ve)=\,&{M}_{z}(\ve)-(1+z^{2}N_{0})
=z^{2}\big((1-\ve^{2})({N}_{\varepsilon z}-N_{0})-\ve^{2}N_{0}\big)
\end{align*}
and so, by $\ve<1$, 
$$
\|{M}^{(1)}_{z}(\ve)\|\le |z|^{2}\left(\| N_{\ve z}-N_{0}\| +\ve^{2}\|N_{0}\|\right)\le
|z|^{2}\left(\| N_{\ve z}-N_{0}\| +\ve\,\|N_{0}\|\right)\,.
$$
Furthermore,
$$
N_{\ve z}-N_{0}=z\int_{0}^{\ve}N_{sz}^{(1)}ds\,,
$$
where 
$$
N^{(1)}_{z}:L^{2}(\Omega)\to L^{2}(\Omega)\,,\qquad N^{(1)}_{z}u(x):=-\frac1{4\pi}\int_{\Omega}\,e^{-z|x-y|}\,u(y)\,dy\,.
$$
Then, by
$$
\int_{\Omega}|N^{(1)}_{z}u(x)|^{2}dx\le\frac{e^{-2\text{Re}(z)\,d_{\Omega}}}{(4\pi)^{2}}\left(\int_{\Omega}|u(y)|\,dy\right)^{\!\! 2}\le \frac{|\Omega|}{(4\pi)^{2}}\,\| u\| ^{2}_{L^{2}(\Omega)}\,,
$$
one obtains
\begin{align*}
\| N_{\ve z}-N_{0}\|\le|z|\int_{0}^{\ve}\| N^{(1)}_{sz}\|\,ds
\le\ve\,\frac{|z|}{4\pi}\,|\Omega|^{1/2}
\,.
\end{align*}
\end{proof}
\begin{lemma}\label{LL} Let $\ve\in (0,1)$ and $z=c+i\gamma$, $c>0$, $\gamma\in \RE$, be such that
\be\label{1/2}
\|(1+z^{2}N_{0})^{-1}\|\,\|M^{(1)}_{z}(\ve)\|\le \frac12\,.
\ee
Then
$$ 
{M}_{z}(\ve)^{-1}=(1+z^{2}N_{0})^{-1}+\Lambda_{z}(\ve)
$$
and
$$
\|\Lambda_{z}(\ve)\|\le 2\ve\,|z|^{2}\left({\frac{|\Omega|}{4\pi}}^{1/2}\!\!|z|+\|N_{0}\|\right)\,\|(1+z^{2}N_{0})^{-1}\|^{2}\,.
$$
\end{lemma}
\begin{proof}  
By 
$$
{M}_{z}(\ve)=(1+z^{2}N_{0})(1+(1+z^{2}N_{0})^{-1}{M}^{(1)}_{z}(\ve))
$$
and by \eqref{1/2}, one gets
$$
{M}_{z}(\ve)^{-1}=\left(\sum_{n=0}^{+\infty}((1+z^{2}N_{0})^{-1}{M}^{(1)}_{z}(\ve))^{n}\right)(1+z^{2}N_{0})^{-1}
$$
and
$$
\|{M}_{z}(\ve)^{-1}\|\le 2\,\|(1+z^{2}N_{0})^{-1}\|\,.
$$
Hence, by Lemma \ref{LM} and by
\be\label{MMN}
\Lambda_{z}(\ve)=-{M}_{z}(\ve)^{-1}M^{(1)}_{z}(\ve)(1+z^{2}N_{0})^{-1}\,,
\ee
one gets
\begin{align*}
&\|\Lambda_{z}(\ve)\|\le \|{M}_{z}(\ve)^{-1}\|\,\|M^{(1)}_{z}(\ve)\|\,\|(1+z^{2}N_{0})^{-1}\|\\
&\,\le 2\ve\,|z|^{2}\left({\frac{|\Omega|}{4\pi}}^{1/2}\!\!|z|+\|N_{0}\|\right)\,\|(1+z^{2}N_{0})^{-1}\|^{2}\,.
\end{align*}
\end{proof}
\begin{lemma}\label{LL1} Let $z=c+i\gamma$, $c>0$, $\gamma\in \RE$ be such that $|\gamma|\le c$. Then, 
$$
\|\Lambda_{z}(\ve)\|\le \ve\,2c^{2}\left({\frac{|\Omega|^{1/2}c}{2\sqrt2\pi}}+\|N_{0}\|\right)\,. 
$$
\end{lemma}
\begin{proof} By $\gamma^{2}\le c^{2}$ one gets Re$(z^{2})\ge 0$. Then, by \eqref{TR}, \eqref{ns}, and \eqref{KR_id_1}, one obtains
$$
\|S_{z}(\ve)^{-1}\|\le \|(1+T(\ve)R_{z^{2}})^{-1}\|\le\frac1{\ve^{2}}\,.
$$
Hence, by \eqref{S_k_id} and by 
$$
\|1_{\Omega}U_{\ve}1^{*}_{\Omega_{\ve}}\|=\|1_{\Omega_{\ve}}U^{*}_{\ve}1^{*}_{\Omega}\|=1\,,
$$
one has
$$
\|M_{z}(\ve)^{-1}\|\le 1\,.
$$
The proof is then concluded by \eqref{MMN}, by Lemmata \ref{LN}, \ref{LM} and by $|z|\le \sqrt2\, c$.
\end{proof}
\begin{lemma}\label{CC} Suppose that $z=c+i\gamma$, $|\gamma|> c>0$. Let
$$
c_{1}:=\frac{1}{4}\left(\frac{|\Omega|^{1/2}}{2\sqrt2\pi}+\|N_{0}\|\right)^{\!\!-1}\,.
$$
If 
\be\label{CC1}
\frac{\ve}{c}\le  c_{1}\,,\qquad
|\gamma|^{4}\le\frac{c}{\ve}\,c_{1}\,,
\ee
then \eqref{1/2} holds true. 
\end{lemma}
\begin{proof} By $|\gamma|>c$, by Lemmata \ref{LM} and \ref{LN}, 
\begin{align*}
\| (1+z^{2}N_{0})^{-1}\|\,\|{M}^{(1)}_{z}(\ve)\|\le &\,
 \ve\,\frac{(c^{2}+\gamma^{2})^{2}}{2c|\gamma|}\left(\frac{|\Omega|^{1/2}}{4\pi}\,(c^{2}+\gamma^{2})^{1/2}+\|N_{0}\|\right)\\
 \le&\,
\ve\,\frac{2|\gamma|^{3}}{c}\left(\frac{|\Omega|^{1/2}}{2\sqrt2\pi}\,|\gamma|+\|N_{0}\|\right) .
\end{align*}
If $|\gamma|\le 1$, then 
$$
\| (1+z^{2}N_{0})^{-1}\|\,\|{M}^{(1)}_{z}(\ve)\|\le \ve\,\frac{2}{c}\left(\frac{|\Omega|^{1/2}}{2\sqrt2\pi}+\|N_{0}\|\right).
$$
If $|\gamma|\ge 1$, then 
$$
\| (1+z^{2}N_{0})^{-1}\|\,\|{M}^{(1)}_{z}(\ve)\|\le \ve\,|\gamma|^{4}\,\frac{2}{c}\left(\frac{|\Omega|^{1/2}}{2\sqrt2\pi}+\|N_{0}\|\right).
$$
Hence, \eqref{CC1} implies \eqref{1/2}.
\end{proof}
Given 
$$
G_{z}(\ve):L^{2}(\Omega)\to L^{2}(\RE^{3})\,,
$$
as in \eqref{gze}, and  
\be\label{Gz}
G_{z}:L^{2}(\Omega)\to L^{2}(\RE^{3})\,,\qquad 
G_{z}u:=\langle1,u\rangle_{L^{2}(\Omega)}\,{\mathcal G}_{z}\,,\qquad {\mathcal G}_{z}(x):=\frac{e^{-z|x|}}{4\pi|x|}\,,
\ee
let 
$$
G^{(1)}_{z}(\ve):L^{2}(\Omega)\to L^{2}(\RE^{3})\,,
$$ 
be such that 
$$
G_{z}(\ve)=G_{z}+G^{(1)}_{z}(\ve)\,.
$$
The next result gives bounds on the norms of such three bounded operators: 
\begin{lemma}\label{LG} 
Let $z=c+i\gamma$, $c>0$, $\gamma\in\RE$. Then,\par\noindent
1)
$$
\| G_{z}\|\le \left(\frac{|\Omega|}{8\pi c}\right)^{\!\!1/2}\,,
$$
\par\noindent
2)
$$
\| G^{(1)}_{z}(\ve)\|\le \ve^{1/2}\, \left(\int_{\Omega}|x|\,dx\right)^{1/2}\,\begin{cases}
2+|z|^{-2}&|\gamma|\le c\\
1+(2c|\gamma|\,)^{-1}(1+|z|^{2})&|\gamma|> c
\end{cases}\,,
$$
\par\noindent
3)
$$
\| G_{z}(\ve)\|\le \,  n_{s}\,|\Omega|\,\ve^{s-3/2}\,
\begin{cases}\left(2+|z|^{-2}\right)^{s/2}|z|^{s-2}&|\gamma|\le c\\
\big(1+(c+|\gamma|\,)^{2})\big)^{s/2}\left(2c|\gamma|\right)^{-1}&|\gamma|> c\end{cases}\,.
$$
Here, $0\le s<\frac32$ and $n_{s}$ denotes the norm of the Sobolev embedding $H^{s}(\RE^{3})\hookrightarrow L^{q}(\RE^{3})$, $q=\frac{6}{3-2s}$.
\end{lemma}
\begin{proof} 1)
$$
\|G_{z}\|^{2}\le |\Omega|\,\|{\mathcal G}_{z}\|^{2}=\frac{|\Omega|}{4\pi}\,\int_{0}^{\infty}e^{-2\text{Re}(z)r}\,dr=\,\frac{|\Omega|}{8\pi\text{Re}(z)}
$$
2) By the continuous embedding of $H^{2}(\Omega)$ into the space of H\"older-continuous functions of order $\frac12$, one gets 
\begin{align*}
\| (G^{*}_{\bar z}(\ve)-G^{*}_{\bar z})u\|^{2}
=&\,\int_{\Omega}\left|\int_{\RE^{3}}({\mathcal G}_{z}(\ve x-y)-{\mathcal G}_{z}(y))u(y)\,dy\right|^{2}dx\\
=&\,\int_{\Omega}\big|R_{z^{2}}u(\ve x)-R_{z^{2}}u(0)\big|^{2}dx\le\ve 
\int_{\Omega}|x|\,dx\,\|R_{z^{2}}u\|^{2}_{H^{2}(\Omega)}\\
=&\,\ve\int_{\Omega}|x|\,dx\ \|(-\Delta+z^{2}+1-z^{2})R_{z^{2}}u\|^{2}\\
\le&\, \ve\int_{\Omega}|x|\,dx\ \left(\|u\|+(1+|z|^{2})\|R_{z^{2}}u\|\right)^{2}\\
\le&\, \ve\int_{\Omega}|x|\,dx\ \begin{cases}
\left(2+|z|^{-2}\right)^{2}\|u\|^{2}&|\gamma|\le c\\
\left(1+(2c|\gamma|\,)^{-1}(1+|z|^{2})\right)^{2}\|u\|^{2}&|\gamma|> c
\end{cases}\,.
\end{align*}
3) By the inequalities, 
\begin{align*}
\| R_{z^{2}}\|\le \begin{cases}|z|^{-2}&|\gamma|\le c\\
\left(2c|\gamma|\right)^{-1}&|\gamma|> c
\end{cases}
\,,\qquad \|(-\Delta+1) R_{z^{2}}\|\le 
\begin{cases}
2+|z|^{-2}&|\gamma|\le c\\
1+(2c|\gamma|\,)^{-1}(1+|z|^{2})&|\gamma|> c
\end{cases}
\,,
\end{align*}
and by interpolation, one obtains
$$
\| R_{z^{2}}\|_{ L^{2}(\RE^{3}), H^{s}(\RE^{3})}\le 
\begin{cases}\left(2+|z|^{-2}\right)^{s/2}|z|^{s-2}&|\gamma|\le c\\
\left(1+(2c|\gamma|\,)^{-1}(1+|z|^{2})\right)^{s/2}\left(2c|\gamma|\right)^{s/2-1}&|\gamma|> c\end{cases}\,,\qquad 0\le s\le 2
\,.
$$
Then, by the Sobolev embedding $$H^{s}(\RE^{3})\hookrightarrow L^{q}(\RE^{3})\,,\qquad q=\frac{6}{3-2s}\,,\quad\quad 0\le s<\frac32\,,
$$
and by H\"older's inequality with $\frac1p=\frac s3$, so that $\frac1{p}+\frac1{q}=\frac12$, one gets
\begin{align*}  
&\| G_{z}(\ve)u\|=\ve^{-3/2}\left(\int_{\RE^{3}}|\chi_{\Omega_{\ve}}(x)R_{z^{2}}u(x)|^{2}
dx\right)^{1/2}\\
\le\,& \ve^{-3/2}|\Omega_{\ve}|^{1/p}\| R_{z^{2}}u\|_{ L^{q}(\RE^{3})}\\
\le\,&n_{s}\,|\Omega|\,\ve^{s-3/2} \| R_{z^{2}}u\|_{ H^{s}(\RE^{3})}\\
\le\,& 
n_{s}\,|\Omega|\,\ve^{s-3/2}\,\| u\|\begin{cases}\left(2+|z|^{-2}\right)^{s/2}|z|^{s-2}&|\gamma|\le c\\
\left(1+(2c|\gamma|\,)^{-1}(1+|z|^{2})\right)^{s/2}\left(2c|\gamma|\right)^{s/2-1}&|\gamma|> c\,.\end{cases}
\end{align*}
\end{proof}
\section{Asymptotic dynamics}
By \eqref{res}, \eqref{Kz} and by the results in the previous section, one has 
$$
R_{z^{2}}(\varepsilon)=R_{z^{2}}-K_{z^{2}}(\ve)=R_{z^{2}}-
(1-\ve^{2})\,\ve\, G_{z}(\ve){M}_{z}(\ve)^{-1}G_{\bar z}(\ve)^{*}\Delta
$$
and
\begin{align}\label{RR}
&G_{z}(\ve){M}_{z}(\ve)^{-1}G_{\bar z}(\ve)^{*}\nonumber\\
=\,&G_{z}(\ve)(1+z^{2}N_{0})^{-1}G_{\bar z}(\ve)^{*}+G_{z}(\ve)\Lambda_{z}(\ve)G_{\bar z}(\ve)^{*}\nonumber\\
=\,&G_{z}(1+z^{2}N_{0})^{-1}G_{\bar z}^{*}+
G_{z}(1+z^{2}N_{0})^{-1}G_{\bar z}^{(1)}(\ve)^{*}\\
+\,&
G_{z}^{(1)}(\ve)(1+z^{2}N_{0})^{-1}G_{\bar z}(\ve)^{*}+
G_{z}(\ve)\Lambda_{z}(\ve)G_{\bar z}(\ve)^{*}\nonumber\,.
\end{align}
Then, combining the estimates provided in Lemmata \ref{LN}-\ref{LG} with Lemma \ref{cp-ep}, one gets
\begin{theorem}\label{T1} Let $\phi$ and $\psi$ be in $H^{6}(\RE^{3})$ and supported outside $\Omega_{\ve}$; let 
$u_{\ve}$ and $u_{\text {\rm free}}$ be the classical solutions of the Cauchy problems 
$$
\begin{cases}
(\ve^{-2}\chi_{\Omega_{\ve}}+\chi_{{\RE^{3}\backslash\Omega_{\ve}}})\partial_{tt} u_{\ve}(t)=\Delta u_{\ve}(t)\\
u_{\ve}(0)=\phi\\
\partial_{t} u_{\ve}(0)=\psi\,.
\end{cases}
\qquad
\begin{cases}
\partial_{tt} u_{\text {\rm free}}(t)=\Delta u_{\text {\rm free}}(t)\\
u_{\text {\rm free}}(0)=\phi\\
\partial_{t}u_{\text {\rm free}}(0)=\psi\,.
\end{cases}
$$
Then, for any $\tau\in(0,\frac1{11})$, $\tau'\in(0,\frac1{9})$ and for any $\ve>0$ such that 
$$
\ve\le\min\left\{1, c_{1}^{4/(3\tau+1)}, c_{1}^{4/(3\tau'+1)}\right\}\,,
$$ 
one has
$$
u_{\ve}(t)=\,u_{\text {\rm free}}(t)+\ve\, (\ve^{2}-1)\,v(t)+r_{\ve}(t)+r'_{\ve}(t)\,,\qquad v(t,x)=H(t-|x|)\,\frac{q(t-|x|)}{4\pi|x|}\,,
$$
where 
$$
\sup_{0<t\le1/\ve^{\tau}}\|r_{\ve}(t)\|\lesssim \,\ve^{1+(1-11\tau)/2}\,\|\Delta^{3}\psi\|\,,
\qquad
\sup_{0<t\le1/\ve^{\tau'}}\|r'_{\ve}(t)\|\lesssim \,\ve^{1+(1-9\tau')/4}\,\|\Delta^{3}\phi\|\,,
$$
\begin{align*}
q(t):=
\sum_{k=1}^{+\infty}|\langle e_{k},1\rangle|^{2}
\int_{|y|<t}\bigg(
\bigg(1-\cos\bigg(\frac{t-|y|}{{\lambda_{k}^{1/2}}}\bigg)\bigg)
\frac{\Delta^{2}\phi(y)}{4\pi|y|}
+\frac1{\lambda_{k}^{1/2}}\,\sin\bigg(\frac{t-|y|}{{\lambda_{k}^{1/2}}}\bigg)
\frac{\Delta\psi(y)}{4\pi|y|}&
\,\bigg)\,dy
\end{align*}
and the $e_{n}$'s and the $\lambda_{n}$'s are the eigenvectors and the eigenvalues of $N_{0}$.\end{theorem}
\begin{proof} By Lemma \ref{cp-ep}, there holds, for any $\phi\in H^{2(m'+1)}(\RE^{3})$, $m'\ge 0$,  and $\psi\in H^{2n'}(\RE^{3})$, $n'\ge 1$,  supported outside $\Omega_{\ve}$,  
$$
u_{\ve}(t)=u_{\text {\rm free}}(t)
-\frac1{2\pi i}\int_{c-i\infty}
^{c+i\infty} e^{ tz} K_{z^{2}}(\varepsilon)\Delta^{n'}\psi\ \frac{dz}{z^{2n'}}
-\frac1{2\pi i}\int_{c'-i\infty}
^{c'+i\infty} e^{ tz} K_{z^{2}}(\varepsilon)\Delta^{m'+1}\phi\ \frac{dz}{z^{2m'+1}}\,,
$$
where $c>0$ and $c'>0$ are arbitrary.
Using Corollary \ref{coro}, one has, for any $1\le n< n'$, 
\begin{align*}
&\frac1{2\pi i}\int_{c-i\infty}
^{c+i\infty} e^{ tz} K_{z^{2}}(\ve)\Delta^{n'}\psi\,\frac{dz}{z^{2n'}}=
\frac1{2\pi i}\int_{c-ic_{1}(c/\ve)^{1/4}}
^{c+ic_{1}(c/\ve)^{1/4}} e^{ tz} K_{z^{2}}(\ve)\Delta^{n}\psi\,\frac{dz}{z^{2n}}+I(t,\ve)
\end{align*}
and, for any $0\le m<m'$,
\begin{align*}
&\frac1{2\pi i}\int_{c'-i\infty}
^{c'+i\infty} e^{ tz} K_{z^{2}}(\ve)\Delta^{m'+1}\phi\,\frac{dz}{z^{2m'+1}}=
\frac1{2\pi i}\int_{c'-ic_{1}(c'/\ve)^{1/4}}
^{c'+ic_{1}(c'/\ve)^{1/4}} e^{ tz} K_{z^{2}}(\ve)\Delta^{m+1}\phi\,\frac{dz}{z^{2m+1}}+I'(t,\ve)\,,
\end{align*}
where
$$
I(t,\ve):=\frac1{2\pi i}\int_{c-i\infty}^{c-ic_{1}(c/\ve)^{1/4}}\!\!\!e^{ tz} K_{z^{2}}(\ve)\Delta^{n'}\psi\,\frac{dz}{z^{2n'}}+
\frac1{2\pi i}\int_{c+ic_{1}(c/\ve)^{1/4}}^{c+i\infty} e^{ tz} K_{z^{2}}(\ve)\Delta^{n'}\psi\,\frac{dz}{z^{2n'}}\,,
$$
$$
I'(t,\ve):=\frac1{2\pi i}\int_{c'-i\infty}^{c'-ic_{1}(c'/\ve)^{1/4}}\!\!\!e^{ tz} K_{z^{2}}(\ve)\Delta^{m'+1}\phi\,\frac{dz}{z^{2m'+1}}+
\frac1{2\pi i}\int_{c'+ic_{1}(c'/\ve)^{1/4}}^{c'+i\infty} e^{ tz} K_{z^{2}}(\ve)\Delta^{m'+1}\phi\,\frac{dz}{z^{2m'+1}}\,.
$$
By \eqref{KE}, whenever 
\be\label{c1c}
c_{1}\left(\frac{c}{\ve}\right)^{1/4}\ge c\,,\qquad c_{1}\left(\frac{c'}{\ve}\right)^{1/4}\ge c'\,,
\ee
one gets
$$
\|I(t,\ve)\|\lesssim \frac{e^{ct}}{c}\,\int_{c_{1}(c/\ve)^{1/4}}^{+\infty}\frac{d\gamma}{\gamma^{2n'+1}}\ \|\Delta^{n'}\psi\|
\lesssim \left(\frac{\ve}{c}\right)^{n'/2}\,\frac{e^{ct}}{c}\ \|\Delta^{n'}\psi\|\,.
$$
$$
\|I'(t,\ve)\|\lesssim \frac{e^{c't}}{c'}\,\int_{c_{1}(c'/\ve)^{1/4}}^{+\infty}\frac{d\gamma}{\gamma^{2m'+2}}\ \|\Delta^{m'+1}\phi\|
\lesssim \left(\frac{\ve}{c'}\right)^{m'/2+1/4}\,\frac{e^{c't}}{c'}\ \|\Delta^{m'+1}\phi\|\,,
$$
Now, we re-write $K_{z^{2}}(\ve)$ by using \eqref{RR} and provide estimates on the various terms. We begin by considering the evolution of the initial datum $\psi$.\par
One has
\begin{align*}
&\,\frac1{2\pi i}\int_{c-i\infty}
^{c+i\infty} e^{ tz} K_{z^{2}}(\ve)\Delta^{n}\psi\,\frac{dz}{z^{2n}}\\
=&\,(1-\ve^{2})\ve\big(v_{1}(t)
+I_{0}(t,\ve)+I_{1}(t,\ve)+I_{1}(t,\ve)^{*}+{I_{2}}(t\ve)\big)+I(t,\ve)\,,
\end{align*}
where
\be\label{v0}
v_{0}(t):=\frac1{2\pi i}\int_{c-i\infty}
^{c+i\infty} e^{ tz} G_{z}(1+z^{2}N_{0})^{-1}G_{\bar z}^{*}\Delta^{n+1}\psi\,\frac{dz}{z^{2n}}
\ee
\begin{align*}
I_{0}(t,\ve):=&\,-\frac1{2\pi i}\int_{c-i\infty}^{c-ic_{1}(c/\ve)^{1/4}}
e^{ tz} G_{z}(1+z^{2}N_{0})^{-1}G_{\bar z}^{*}\Delta^{n+1}\psi\,\frac{dz}{z^{2n}}\\
&-
\frac1{2\pi i}\int_{c+ic_{1}(c/\ve)^{1/4}}^{c+i\infty} 
e^{ tz} G_{z}(1+z^{2}N_{0})^{-1}G_{\bar z}^{*}\Delta^{n+1}\psi\,\frac{dz}{z^{2n}}\,,
\end{align*}
$$
I_{1}(t,\ve):=\frac1{2\pi i}\int_{c-ic_{1}(c/\ve)^{1/4}}
^{c+ic_{1}(c/\ve)^{1/4}} e^{ tz} G_{z}(1+z^{2}N_{0})^{-1}G_{\bar z}^{(1)}(\ve)^{*}\Delta^{n+1}\psi\,\frac{dz}{z^{2n}}\,,
$$
$$
I_{1}(t,\ve)^{*}:=\frac1{2\pi i}\int_{c-ic_{1}(c/\ve)^{1/4}}
^{c+ic_{1}(c/\ve)^{1/4}} e^{ tz} G_{z}^{(1)}(\ve)(1+z^{2}N_{0})^{-1}G_{\bar z}^{*}\Delta^{n+1}\psi\,\frac{dz}{z^{2n}}\,,
$$
$$
I_{2}(t,\ve):=\frac1{2\pi i}\int_{c-ic_{1}(c/\ve)^{1/4}}
^{c+ic_{1}(c/\ve)^{1/4}} e^{ tz} G_{z}(\ve)\Lambda_{z}(\ve)G_{z}(\ve)^{*}\Delta^{n+1}\psi\,\frac{dz}{z^{2n}}\,.
$$
By Lemmata \ref{LN} and \ref{LG}, whenever $|\gamma|>c$,
$$
\| G_{z}(1+z^{2}N_{0})^{-1}G_{\bar z}^{*}\|\le\,\|{G}_{z}\|\,\|{G}_{\bar z}\|\,\|(1+z^{2}N_{0})^{-1}\|
\lesssim 
\frac1{c^{2}}\,\frac{|z|^{2}}{|\gamma|\ }
$$
and so, 
$$
\|I_{0}(t,\ve)\|\lesssim \frac{e^{ct}}{c^{2}}\,\int_{c_{1}(c/\ve)^{1/4}}^{+\infty}\frac{d\gamma}{\gamma^{2n-1}}\ \|\Delta^{n+1}\psi\|
\lesssim \left(\frac{\ve}{c}\right)^{(n-1)/2}\,\frac{e^{ct}}{c^{2}}\,\|\Delta^{n+1}\psi\|\,.
$$
By Lemmata \ref{I1} and \ref{I2} in the Appendix,
$$
\|I_{1}(t,\ve)\|\lesssim\ve^{1/2}\,\frac{e^{tc}}{c^{2n+3/2}}\ \|\Delta^{n+1}\psi\|\,,
\qquad \|I_{2}(t,\ve)\|\lesssim \ve^{1/2}\,\frac{e^{tc}}{c^{2n+1}}\ \|\Delta^{n+1}\psi\|\,.
$$
Then, the estimate for   $\|I_{1}(t,\ve)^{*}\|$ and $\|I_{2}(t,\ve)^{*}\|$ are the same as the ones for $\|I_{1}(t,\ve)\|$ and $\|I_{2}(t,\ve)\|$.\par
Putting the previous estimates together, one gets
\begin{align*}
&\,{\mathcal I}(t,\ve):=\left\|(1-\ve^{2})\ve
\big(I_{0}(t,\ve)+I_{1}(t,\ve)+I_{1}(t,\ve)^{*}+{I_{2}}(t,\ve)\big)+I(t,\ve)
\right\|\\
\lesssim&\,\ve\bigg(\left(\frac{\ve}{c}\right)^{n/2-1/2}\,\frac{e^{ct}}{c^{2}}+\ve^{1/2}\left(\frac{e^{tc}}{c^{2n+3/2}}+\frac{e^{tc}}{c^{2n+1}}\right)\bigg)
\|\Delta^{n+1}\psi\|+\left(\frac{\ve}{c}\right)^{n'/2}\,\frac{e^{ct}}{c}\,\|\Delta^{n'}\psi\|\\
\lesssim&\,\ve\bigg(\left(\frac{\ve}{c}\right)^{n/2-1/2}\,\frac{e^{ct}}{c^{2}}+\ve^{1/2}\,\frac{e^{tc}}{c^{2n+3/2}}\bigg)\,\|\Delta^{n+1}\psi\|+\left(\frac{\ve}{c}\right)^{n'/2}\,\frac{e^{ct}}{c}\,\|\Delta^{n'}\psi\|\,.
\end{align*}
Then, taking $c=\ve^{\tau}$, $\tau>0$, and $t\le\ve^{-\tau}$, 
\begin{align*}
{\mathcal I}(t,\ve)\lesssim&\,\ve\left(\ve^{(1-\tau)(n/2-1/2)-2\tau}+\ve^{1/2-(2n+3/2)\tau}\right)\|\Delta^{n+1}\psi\|+\ve^{(1-\tau)n'/2-\tau}\,\|\Delta^{n'}\psi\|
\end{align*}
and, by \eqref{c1c}, $\ve\le c_{1}^{4/(3\tau+1)}$.
We choose $\tau$ such that 
$$
(1-\tau)\left(\frac{n}2-\frac12\right)-2\tau>0\,,\quad \frac12-\left(2n+\frac32\right)\tau>0\,,\quad (1-\tau)\,\frac{n'}2-\tau>1\,,
$$
equivalently,
$$
\tau<\frac{n-1}{n+3}\,,\quad \tau<\frac1{4n+3}\,,\quad \tau<\frac{n'-2}{n'+2}
$$
Taking $n=2$, $n'=3$, one obtains the bound $\tau<\frac1{11}$ and
 \begin{align*}
{\mathcal I}(t,\ve)\lesssim&\,\ve\left(\ve^{(1-3\tau)/2}+\ve^{(1-11\tau)/2}\right)\|\Delta^{3}\psi\|+\ve^{(3-5\tau)/2}\,\|\Delta^{3}\psi\|
\lesssim\ve^{1+(1-11\tau)/2}
\|\Delta^{3}\psi\|\,.
\end{align*}
Proceeding as above, one has, as regards the evolution of the initial datum $\phi$,
\begin{align*}
&\,\frac1{2\pi i}\int_{c'-i\infty}
^{c'+i\infty} e^{ tz} K_{z^{2}}(\ve)\Delta^{m+1}\phi\,\frac{dz}{z^{2m+1}}\\
&\,=(1-\ve^{2})\ve\big(v'_{0}(t)+I'_{0}(t,\ve)+I'_{1}(t,\ve)+I'_{1}(t,\ve)^{*}+{I'_{2}}(t,\ve)\big)+I'(t,\ve)\,,
\end{align*}
where
\be\label{v0p}
v'_{0}(t):=\frac1{2\pi i}\int_{c'-i\infty}
^{c'+i\infty} e^{ tz} G_{z}(1+z^{2}N_{0})^{-1}G_{\bar z}^{*}\Delta^{m+2}\phi\,\frac{dz}{z^{2m+1}}\,,
\ee
\begin{align*}
I'_{0}(t,\ve):=&\,-\frac1{2\pi i}\int_{c'-i\infty}^{c'-ic_{1}(c'/\ve)^{1/4}}
e^{ tz} G_{z}(1+z^{2}N_{0})^{-1}G_{\bar z}^{*}\Delta^{m+2}\phi\,\frac{dz}{z^{2m+1}}\\
&-
\frac1{2\pi i}\int_{c'+ic_{1}(c'/\ve)^{1/4}}^{c'+i\infty} 
e^{ tz} G_{z}(1+z^{2}N_{0})^{-1}G_{\bar z}^{*}\Delta^{m+2}\phi\,\frac{dz}{z^{2m+1}}
\,,
\end{align*}
$$
I'_{1}(t,\ve):=\frac1{2\pi i}\int_{c'-ic_{1}(c'/\ve)^{1/4}}
^{c'+ic_{1}(c'/\ve)^{1/4}} e^{ tz} G_{z}(1+z^{2}N_{0})^{-1}G_{\bar z}^{(1)}(\ve)^{*}\Delta^{m+2}\phi\,\frac{dz}{z^{2m+1}}\,,
$$
$$
I'_{1}(t,\ve)^{*}:=\frac1{2\pi i}\int_{c'-ic_{1}(c'/\ve)^{1/4}}
^{c'+ic_{1}(c'/\ve)^{1/4}} e^{ tz} G_{z}^{(1)}(\ve)(1+z^{2}N_{0})^{-1}G_{\bar z}^{*}\Delta^{m+2}\phi\,\frac{dz}{z^{2m+1}}\,,
$$
$$
I'_{2}(t,\ve):=\frac1{2\pi i}\int_{c'-ic_{1}(c'/\ve)^{1/4}}
^{c'+ic_{1}(c'/\ve)^{1/4}} e^{ tz} G_{z}(\ve)\Lambda_{z}(\ve)G_{z}(\ve)^{*}\Delta^{m+2}\phi\,\frac{dz}{z^{2m+1}}\,.
$$
One has
$$
\|I'_{0}(t,\ve)\|\lesssim \frac{e^{c't}}{c'^{2}}\,\int_{c_{1}(c'/\ve)^{1/4}}^{+\infty}\frac{d\gamma}{\gamma^{2m}}\ \|\Delta^{m+2}\phi\|
\lesssim \left(\frac{\ve}{c'}\right)^{m/2-1/4}\,\frac{e^{c't}}{c'^{2}}\,\|\Delta^{m+2}\phi\|\,.
$$
As regard the estimates for $\|I'_{1}(t,\ve)\|$ and
$\|I'_{2}(t,\ve)\|$, it suffices to proceed as in the proofs of Lemmata \ref{I1} and \ref{I2}, with the foresight to take into account the different scaling properties of $(c'^{2}+\gamma^{2})^{-(m+1/2)}$ with respect to $(c^{2}+\gamma^{2})^{-n}$; this means that in the estimates, after the scaling $\gamma\mapsto \gamma/c'$, the exponent $2n$ has to be replaced by $2m+1$. Therefore, one gets
$$
\|I'_{1}(t,\ve)\|\lesssim\,\ve^{1/2}\,\frac{e^{tc}}{c^{2m+5/2}}\ \|\Delta^{m+2}\phi\|\,,
\qquad
\|I'_{2}(t,\ve)\|\lesssim \,\ve^{1/2}\,\frac{e^{tc}}{c^{2(m+1)}}\ \|\Delta^{m+2}\phi\|
$$
and the estimates for   $\|I'_{1}(t,\ve)^{*}\|$ and $\|I'_{2}(t,\ve)^{*}\|$ are the same as the ones for 
$\|I'_{1}(t,\ve)\|$ and $\|I'_{2}(t,\ve)\|$.\par
Putting such estimates together, one gets
\begin{align*}
&\,{\mathcal I'}(t,\ve):=\left\|(1-\ve^{2})\ve
\big(I'_{0}(t,\ve)+I'_{1}(t,\ve)+I'_{1}(t,\ve)^{*}+{I'_{2}}(t,\ve)\big)+I'(t,\ve)
\right\|\\
\lesssim&\,\ve\bigg(\left(\frac{\ve}{c'}\right)^{m/2-1/4}\,\frac{e^{c't}}{c'^{2}}+\ve^{1/2}\left(\frac{e^{tc'}}{c'^{2m+5/2}}+\frac{e^{tc'}}{c'^{2m+2}}\right)\bigg)
\|\Delta^{m+2}\phi\|+\left(\frac{\ve}{c'}\right)^{m'/2+1/4}\,\frac{e^{c't}}{c'}\,
\|\Delta^{m'+1}\phi\|\\
\lesssim&\,\ve\bigg(\left(\frac{\ve}{c'}\right)^{m/2-1/4}\,\frac{e^{c't}}{c'^{2}}+\ve^{1/2}\,\frac{e^{tc'}}{c'^{2m+5/2}}\bigg)\|\Delta^{m+2}\phi\|
+\left(\frac{\ve}{c'}\right)^{m'/2+1/4}\,\frac{e^{c't}}{c'}\,\|\Delta^{m'+1}\phi\|\,.
\end{align*}
Then, taking $c'=\ve^{\tau'}$, $\tau'>0$, and $t\le\ve^{-\tau'}$, 
\begin{align*}
{\mathcal I'}(t,\ve)\lesssim&\,\ve\left(\ve^{(1-\tau')(m/2-1/4)-2\tau'}+\ve^{1/2-(2m+5/2)\tau'}\right)\|\Delta^{m+2}\phi\|\\
&+\ve^{(1-\tau')(m'/2+1/4)-\tau'}\,\|\Delta^{m'+1}\phi\|
\end{align*}
and, by \eqref{c1c}, $\ve\le c_{1}^{4/(3\tau'+1)}$. We choose $\tau'$ such that 
$$
(1-\tau')\left(\frac{m}2-\frac14\right)-2\tau'>0\,,\quad \frac12-\left(2m+\frac52\right)\tau>0\,,\quad (1-\tau')\left(\frac{m'}2+\frac14\right)-\tau'>1\,,
$$
equivalently,
$$
\tau'<\frac{2m-1}{2m+7}\,,\quad \tau'<\frac1{4m+5}\,,\quad \tau'<\frac{2m'-3}{2m'+5}\,.
$$
Taking $m=1$ and $m'=2$, one obtains the bound $\tau'<\frac1{9}$ and
 \begin{align*}
{\mathcal I'}(t,\ve)\lesssim&\,\ve\left(\ve^{(1-9\tau')/4}+\ve^{(1-9\tau')/2}\right)
\|\Delta^{3}\phi\|
+\ve^{(5-9\tau')/4}\|\Delta^{3}\phi\|\lesssim\ve^{1+(1-9\tau')/4}
\|\Delta^{3}\phi\|\,.
\end{align*}
Thus,
$$
u(t)=u_{\rm free}(t)+\ve(\ve^{2}-1)v(t)+r_{\ve}(t)+r'_{\ve}(t)\,,
$$
$$v(t):=v_{0}(t)+v'_{0}(t)\,,\quad \|r_{\ve}(t)\|\le{\mathcal I}(t,\ve)\,,\quad \|r'_{\ve}(t)\|\le{\mathcal I'}(t,\ve)\,.
$$
By \eqref{v0}, with $n=2$, \eqref{v0p}, with $m=1$, and by \eqref{Gz}, 
\begin{align*}
v(t)=&\,\frac1{2\pi i}\int_{c-i\infty}
^{c+i\infty} e^{ tz} {\mathcal G}_{z}\langle1,(1+z^{2}N_{0})^{-1}1\rangle\langle{\mathcal G}_{\bar z},\Delta^{3}\psi\rangle\,\frac{dz}{z^{4}}
\\+&\,
\frac1{2\pi i}\int_{c'-i\infty}
^{c'+i\infty} e^{ tz} \langle1,(1+z^{2}N_{0})^{-1}1\rangle\langle{\mathcal G}_{\bar z},\Delta^{3}\phi\rangle
\,\frac{dz}{z^{3}}\,.
\end{align*}
By the spectral resolution of $N_{0}$ (see \eqref{sp-res}), the linear operator ${\mathcal G}_{z}\langle1,(1+z^{2}N_{0})^{-1}1\rangle\langle{\mathcal G}_{\bar z},\cdot\rangle$ has the kernel 
$$
K_{z}(x,y)=\sum_{k=1}^{+\infty}\frac{|\langle e_{k},1\rangle|^{2}}{1+z^{2}\lambda_{k}}\ {\mathcal G}_{z}(x){\mathcal G}_{z}(y)=\sum_{k=1}^{+\infty}\frac{|\langle e_{k},1\rangle|^{2}}{1+z^{2}\lambda_{k}}\ \frac{e^{-z(|x|+|y|)}}{(4\pi)^{2}|x|\,|y|}\,.
$$
Let
$$
K^{(\ell)}_{k,x,y}(z):=\frac1{z^{\ell}}\,\frac{e^{-z(|x|+|y|)}}{1+z^{2}\lambda_{k}}\ ,\qquad \ell\ge 0\, .
$$
Then,
$$
[{\mathcal L}^{-1}K^{(0)}_{k,x,y}](t)
=H(t-|x|-|y|)\,\frac1{{\lambda_{k}^{1/2}}}\,\sin\bigg(\frac{t-|x|-|y|}{{\lambda_{k}^{1/2}}}\bigg)
$$
and, by the Laplace transform properties,
$$
[{\mathcal L}^{-1}K^{(\ell+1)}_{k,x,y}](t)=\int_{0}^{t}[{\mathcal L}^{-1}K^{(\ell)}_{k,x,y}](s)\,ds\,.
$$
Hence,
$$
[{\mathcal L}^{-1}K^{(\ell)}_{k,x,y}](t)=H(t-|x|-|y|)\,\lambda_{k}^{(\ell-1)/2}f^{(m)}_{k,x,y}\left(\frac{t-|x|-|y|}{{\lambda_{k}^{1/2}}}\right)\,,
$$
$$
f^{(\ell+1)}_{k,x,y}(t)=\int_{0}^{t}f^{(\ell)}_{k,x,y}(s)\,ds\,,\qquad f^{(0)}_{k,x,y}(t)=\sin t\,.
$$
The latter recursive relations give
$$
f^{(3)}_{k,x,y}(t)=\frac{t^{2}}{2}-1+\cos t\,,\qquad f^{(4)}_{k,x,y}(t)=\frac{t^{3}}{6}-t+\sin t\,.
$$
Therefore,
\begin{align*}
&\,[{\mathcal L}^{-1}K^{(4)}_{k,x,y}](t)\\=&\,H(t-|x|-|y|)\,\lambda_{k}^{3/2}\left(
\frac{(t-|x|-|y|)^{3}}{6\lambda_{k}^{3/2}}-\frac{t-|x|-|y|}{\lambda_{k}^{1/2}}+\sin\bigg(\frac{t-|x|-|y|}{{\lambda_{k}^{1/2}}}\bigg)\right)\,,
\end{align*}
\begin{align*}
[{\mathcal L}^{-1}K^{(3)}_{k,x,y}](t)=H(t-|x|-|y|)\,\lambda_{k}\left(
\frac{(t-|x|-|y|)^{2}}{2\lambda_{k}}-1+\cos\bigg(\frac{t-|x|-|y|}{{\lambda_{k}^{1/2}}}\bigg)\right)\ .
\end{align*}
Let either $\ell=3$ or $\ell=4$. Then, for any $\varphi\in H^{6}(\RE^{3})$, 
\begin{align*}
&\int_{c-i\infty}
^{c+i\infty}\int_{\RE^{3}}\sum_{k=1}^{+\infty}\left|e^{ tz} \frac{|\langle e_{k},1\rangle|^{2}}{1+z^{2}\lambda_{k}}\ {\mathcal G}_{z}(x){\mathcal G}_{z}(y)
 \,\Delta^{3}\varphi(y)\,\right|\,dy\,\frac{dz}{z^{\ell}}\\
\lesssim&\, \frac{e^{ c(t-|x|)}}{|x|}
\int_{-\infty}
^{+\infty} \sum_{k=1}^{+\infty}\frac{|\langle e_{k},1\rangle|^{2}}{|1+(c+i\gamma)^{2}\lambda_{k}|}\,
\frac{d\gamma}{(c^{2}+\gamma^{2})^{\ell/2}} \int_{\RE^{3}}e^{-c|y|}\,|\Delta^{3}\varphi(y)|\,\frac{dy}{|y|}\\
%\end{align*}
%\begin{align*}
\lesssim &\, \frac{e^{ c(t-|x|)}}{|x|}\sum_{k=1}^{+\infty}|\langle e_{k},1\rangle|^{2}
\int_{-\infty}
^{+\infty} \sup_{0\le\lambda\le \|N_{0}\|}\frac1{|1+(c+i\gamma)^{2}\lambda|}\,
\frac{d\gamma}{(c^{2}+\gamma^{2})^{\ell/2}} \,\times\\
&\times\int_{\RE^{3}}e^{-c|y|}\,|\Delta^{3}\varphi(y)|\,\frac{dy}{|y|}\\
\lesssim&\, \frac{e^{ c(t-|x|)}}{|x|}\,|\Omega|\int_{-\infty}
^{+\infty} \max\left\{1,\frac{c^{2}+\gamma^{2}}{2c|\gamma|}\right\}
\frac{d\gamma}{(c^{2}+\gamma^{2})^{\ell/2}}\int_{\RE^{3}}e^{-c|y|}\,|\Delta^{3}\varphi(y)|\,\frac{dy}{|y|}
<+\infty\,.
\end{align*}
%\vfill\eject
Hence, by Fubini's theorem, one has
\begin{align*}
&\,\frac1{2\pi i}\int_{c-i\infty}
^{c+i\infty} e^{ tz} {\mathcal G}_{z}(x)\langle1,(1+z^{2}N_{0})^{-1}1\rangle\langle{\mathcal G}_{\bar z},\Delta^{3}\psi\rangle\,\frac{dz}{z^{4}}\\
%\end{align*}
%\begin{align*}
=&\,\frac1{4\pi|x|}\,\sum_{k=1}^{+\infty}|\langle e_{k},1\rangle|^{2}\,\frac1{2\pi i}\int_{c-i\infty}
^{c+i\infty} e^{ tz} \left(
 \int_{\RE^{3}}K^{(4)}_{k,x,y}(z)\,\frac{\Delta^{3}\psi(y)}{4\pi|y|}\,dy\right)dz
\\
=&\,\frac1{4\pi|x|}\,\sum_{k=1}^{+\infty}|\langle e_{k},1\rangle|^{2}\int_{\RE^{3}}[{\mathcal L}^{-1}K^{(4)}_{k,x,y}](t) \,\frac{\Delta^{3}\psi(y)}{4\pi|y|}
\,dy
\\
=&\,\frac1{4\pi|x|}\,
\sum_{k=1}^{+\infty}|\langle e_{k},1\rangle|^{2}
\int_{\RE^{3}}H(t-|x|-|y|)\,\lambda_{k}^{3/2}\bigg(\,
\frac{(t-|x|-|y|)^{3}}{6\lambda_{k}^{3/2}}\\
&-\frac{t-|x|-|y|}{\lambda_{k}^{1/2}}+\sin\bigg(\frac{t-|x|-|y|}{{\lambda_{k}^{1/2}}}\bigg)
\bigg)
\frac{\Delta^{3}\psi(y)}{4\pi|y|}\,dy\\
=&\,\frac{H(t-|x|)}{4\pi|x|}\,
\sum_{k=1}^{+\infty}|\langle e_{k},1\rangle|^{2}\lambda_{k}^{3/2}
\int_{|y|<t-|x|}\bigg(\,
\frac{(t-|x|-|y|)^{3}}{6\lambda_{k}^{3/2}}\\
&-\frac{t-|x|-|y|}{\lambda_{k}^{1/2}}+\sin\bigg(\frac{t-|x|-|y|}{{\lambda_{k}^{1/2}}}\bigg)
\bigg)
\frac{\Delta^{3}\psi(y)}{4\pi|y|}\,dy
\,.
\end{align*}
%\vfill\eject
Analogously,
\begin{align*}
&\,\frac1{2\pi i}\int_{c-i\infty}
^{c+i\infty} e^{ tz} {\mathcal G}_{z}(x)\langle1,(1+z^{2}N_{0})^{-1}1\rangle\langle{\mathcal G}_{\bar z},\Delta^{3}\phi\rangle\,\frac{dz}{z^{3}}\\
\end{align*}
\begin{align*}
=&\,\frac1{4\pi|x|}\,\sum_{k=1}^{+\infty}|\langle e_{k},1\rangle|^{2}\,\frac1{2\pi i}\int_{c-i\infty}
^{c+i\infty} e^{ tz} \left(
 \int_{\RE^{3}}K^{(3)}_{k,x,y}(z)\,\frac{\Delta^{3}\phi(y)}{4\pi|y|}\,dy\right)dz
\\
=&\,\frac1{4\pi|x|}\,\sum_{k=1}^{+\infty}|\langle e_{k},1\rangle|^{2}
\int_{\RE^{3}}[{\mathcal L}^{-1}K^{(3)}_{k,x,y}](t) \,\frac{\Delta^{3}\phi(y)}{4\pi|y|}
\,dy
\\
%\end{align*}
%\begin{align*}
=&\,\frac1{4\pi|x|}\,
\sum_{k=1}^{+\infty}|\langle e_{k},1\rangle|^{2}
\int_{\RE^{3}}H(t-|x|-|y|)\,\lambda_{k}\bigg(\,
\frac{(t-|x|-|y|)^{2}}{2\lambda_{k}}\\&-1+\cos\bigg(\frac{t-|x|-|y|}{{\lambda_{k}^{1/2}}}\bigg)\bigg)
\frac{\Delta^{3}\phi(y)}{4\pi|y|}\,dy\\
=&\,\frac{H(t-|x|)}{4\pi|x|}\,
\sum_{k=1}^{+\infty}|\langle e_{k},1\rangle|^{2}\lambda_{k}
\int_{|y|<t-|x|}\bigg(\,
\frac{(t-|x|-|y|)^{2}}{2\lambda_{k}}\\&-1+\cos\bigg(\frac{t-|x|-|y|}{{\lambda_{k}^{1/2}}}\bigg)\bigg)
\frac{\Delta^{3}\phi(y)}{4\pi|y|}\,dy
\,.
\end{align*}
The proof is then concluded by Lemma \ref{green}.
\end{proof}
As regards the inhomogeneous Cauchy problem, one has the following 
\begin{theorem}\label{inhom} Let $f\in L_{loc}^{1}(\RE_{+},H^{6}(\RE^{3}))$ be such that $f(t)$ is supported outside $\Omega_{\ve}$ for any $t\in\RE_{+}$; let 
$\widehat u_{\ve}$ and $\widehat u_{\text {\rm free}}$ be the mild (classical if furthermore $f\in{\mathcal C}(\RE_{+}, H^{2}(\RE^{3}))\,$) solutions of the inhomogeneous Cauchy problems 
$$
\begin{cases}
(\ve^{-2}\chi_{\Omega_{\ve}}+\chi_{{\RE^{3}\backslash\Omega_{\ve}}})\partial_{tt} \widehat u_{\ve}(t)=\widehat u_{\ve}(t)+f(t)\\
\widehat u_{\ve}(0)=0\\
\partial_{t} \widehat u_{\ve}(0)=0\,.
\end{cases}
\qquad
\begin{cases}
\partial_{tt} \widehat u_{\text {\rm free}}(t)=\Delta \widehat u_{\text {\rm free}}(t)+f(t)\\
\widehat u_{\text {\rm free}}(0)=0\\
\partial_{t}\widehat u_{\text {\rm free}}(0)=0\,.
\end{cases}
$$
If   
$$
\|\Delta^{3}f(t)\|\lesssim t^{\alpha}\,, \quad t\gg 1\,,\qquad\alpha\ge -1\,,
$$
then, for any $\tau\in\left(0,\frac1{11+2(\alpha+1)}\right)$ and for any $\ve>0$ such that 
$$
\ve\le\min\left\{1, c_{1}^{4/(3\tau+1)}\right\}\,,
$$ 
one has 
$$
\widehat u_{\ve}(t)=\widehat u_{\text {\rm free}}(t)+\ve\, (\ve^{2}-1)\,\widehat v(t)+\widehat r_{\ve}(t)\,,\qquad \widehat v(t,x)
=H(t-|x|)\,\frac{\widehat q(t-|x|)}{4\pi|x|}\,,
$$
where 
$$
\sup_{0<t\le1/\ve^{\tau}}\|\widehat r_{\ve}(t)\|\lesssim\begin{cases}\ve^{1+(1-11\tau)/2}\ln\ve^{-1}&\alpha=-1\\
\ve^{1+(1-(11+2(\alpha+1)\tau)/2}&\alpha>-1\,,
\end{cases}
$$
\begin{align*}
\widehat q(t)\!:=\!\sum_{k=1}^{+\infty}|\langle e_{k},1\rangle|^{2}\,\frac1{
\lambda_{k}^{1/2}}
\int_{0}^{t}\!\!&\int_{|y|<t-s}\!\!\sin\bigg(\frac{t-s-|y|}{{\lambda_{k}^{1/2}}}\bigg)
\frac{\Delta f(s,y)}{4\pi|y|}\,dy\,ds
\end{align*}
and the $e_{n}$'s and the $\lambda_{n}$'s are the eigenvectors and the eigenvalues of $N_{0}$.\par
If $f\in L^{1}(\RE_{+},H^{6}(\RE^{3}))$, then, for any $\tau\in(0,\frac1{11})$ and for any  $\ve>0$ as above, 
$$
\sup_{0<t\le1/\ve^{\tau}}\|\widehat r_{\ve}(t)\|\lesssim
\ve^{1+(1-11\tau)/2}
\,.
$$
\end{theorem}
\begin{proof} By \cite[Theorem 4.1 and Lemma 4.2 in Section II.4]{Fatt} one gets the existence of the solutions with the stated regularity and 
$$
\widehat u_{\text {\rm free}}(t)=\int_{0}^{t}(-\Delta)^{-1/2}\sin((t-s)(-\Delta)^{1/2})f(s)\,ds\,,
$$
$$
\widehat u_{\ve}(t)=\int_{0}^{t}(-A(\ve))^{-1/2}\sin((t-s)(-A(\ve))^{1/2})f(s)\,ds\,.
$$
Therefore, by the same reasonings as in the proof of Theorem \ref{T1}, 
\begin{align*}
&\widehat u_{\ve}(t)=\widehat u_{\text {\rm free}}(t)\\
&\ve\,(\ve^{2}-1)\,\int_{0}^{t}\bigg(
\frac1{2\pi i}\int_{c-i\infty}
^{c+i\infty} e^{ (t-s)z} {\mathcal G}_{z}\langle1,(1+z^{2}N_{0})^{-1}1\rangle\langle{\mathcal G}_{\bar z},\Delta^{3}f(s)\rangle\,\frac{dz}{z^{4}}\bigg)\,ds+\widehat r_{\ve}(t)\,,
\end{align*}
where 
\begin{align*}
\sup_{0<t\le1/\ve^{\tau}}\|\widehat r_{\ve}(t)\|
\lesssim&\,\ve^{1+(1-11\tau)/2}\int_{0}^{1/\ve^{\tau}}\|\Delta^{3} f(s)\|\,ds\,.
\end{align*}
This gives $$\sup_{0<t\le1/\ve^{\tau}}\|\widehat r_{\ve}(t)\|
\lesssim\ve^{1+(1-11\tau)/2}$$
whenever $f\in L^{1}(\RE_{+},H^{6}(\RE^{3}))$. If $\|\Delta^{3}f(t)\|\lesssim t^{\alpha}$, $t\gg 1$, then
$$
\int_{0}^{1/\ve^{\tau}}\|\Delta^{3} f(s)\|\,ds\lesssim
\begin{cases}
1+\ln\ve^{-1}&\alpha=-1\\
1+\ve^{-\tau(\alpha+1)}&\alpha>-1\,,
\end{cases}
$$
and so in this case
$$\sup_{0<t\le1/\ve^{\tau}}\|\widehat r_{\ve}(t)\|
\lesssim\begin{cases}\ve^{1+(1-11\tau)/2}\ln\ve^{-1}&\alpha=-1\\
\ve^{1+(1-11\tau)/2-\tau(\alpha+1)}&\alpha>-1\,.\end{cases}
$$
Similarly to the proof of Theorem \ref{T1},
\begin{align*}
&\frac1{4\pi|x|}\int_{0}^{t}\int_{c-i\infty}
^{c+i\infty}\int_{\RE^{3}}\sum_{k=1}^{+\infty}\left|e^{ (t-s)z} |\langle e_{k},1\rangle|^{2}
 K^{(4)}_{k,x,y}(z)\,\frac{\Delta^{3} f(s,y)}{4\pi|y|}\,\right|\,dy\,dz\,ds\\
\lesssim&\, \frac{e^{-c|x|}}{|x|}
\int_{-\infty}
^{+\infty} \sum_{k=1}^{+\infty}\frac{|\langle e_{k},1\rangle|^{2}}{|1+(c+i\gamma)^{2}\lambda_{k}|}\,
\frac{d\gamma}{(c^{2}+\gamma^{2})^{3}} \int_{\RE^{3}}e^{-c|y|}\,\int_{0}^{t}e^{ c(t-s)}|\Delta^{3} f(s,y)|\,ds\,\frac{dy}{|y|}\\
%\end{align*}
%\begin{align*}
\lesssim&\, \frac{e^{-c|x|}}{|x|}\sum_{k=1}^{+\infty}|\langle e_{k},1\rangle|^{2}
\int_{-\infty}
^{+\infty} \max\left\{1,\frac{c^{2}+\gamma^{2}}{2c\gamma}\right\}\frac{d\gamma}{(c^{2}+\gamma^{2})^{3}} \,\times\\
&\times\int_{\RE^{3}}e^{-c|y|}\,\int_{0}^{t}|\Delta^{3}f(s,y)|\,ds\,\frac{dy}{|y|}<+\infty
\end{align*}
and so, by Fubini's theorem, 
%\vfill\eject
\begin{align*}
&\,\int_{0}^{t}\frac1{2\pi i}\int_{c-i\infty}
^{c+i\infty} e^{ (t-s)z} {\mathcal G}_{z}(x)\langle1,(1+z^{2}N_{0})^{-1}1\rangle\langle{\mathcal G}_{\bar z},\Delta^{3}f(s)\rangle\,\frac{dz}{z^{4}}\,ds\\
=&\,\int_{0}^{t}\frac1{4\pi|x|}\,\sum_{k=1}^{+\infty}|\langle e_{k},1\rangle|^{2}\,\frac1{2\pi i}\int_{c-i\infty}
^{c+i\infty} e^{ (t-s)z} \left(
 \int_{\RE^{3}}K^{(4)}_{k,x,y}(z)\,\frac{\Delta^{3}f(s,y)}{4\pi|y|}\,dy\right)dz\,ds
\\
=&\,\frac1{4\pi|x|}\,\sum_{k=1}^{+\infty}|\langle e_{k},1\rangle|^{2}\int_{0}^{t}\int_{\RE^{3}}[{\mathcal L}^{-1}K^{(4)}_{k,x,y}](t-s) \,\frac{\Delta^{3}f(s,y)}{4\pi|y|}
\,dy\,ds\\
%\end{align*}
%\begin{align*}
=&\,\frac1{4\pi|x|}\,
\sum_{k=1}^{+\infty}|\langle e_{k},1\rangle|^{2}
\int_{0}^{t}\int_{\RE^{3}}H(t-s-|x|-|y|)\,\lambda_{k}^{3/2}
\bigg(\,
\frac{(t-s-|x|-|y|)^{3}}{6\lambda_{k}^{3/2}}\\
&-\frac{t-s-|x|-|y|}{\lambda_{k}^{1/2}}+\sin\bigg(\frac{t-s-|x|-|y|}{{\lambda_{k}^{1/2}}}\bigg)
\bigg)
\frac{\Delta^{3}f(s,y)}{4\pi|y|}\,dy\,ds\\
=&\,\frac{H(t-|x|)}{4\pi|x|}\,
\sum_{k=1}^{+\infty}|\langle e_{k},1\rangle|^{2}
\lambda_{k}^{3/2}
\int_{0}^{t-|x|}\int_{|y|<t-s-|x|}\bigg(\,
\frac{(t-s-|x|-|y|)^{3}}{6\lambda_{k}^{3/2}}\\
&-\frac{t-s-|x|-|y|}{\lambda_{k}^{1/2}}+\sin\bigg(\frac{t-s-|x|-|y|}{{\lambda_{k}^{1/2}}}\bigg)
\bigg)
\frac{\Delta^{3}f(s,y)}{4\pi|y|}\,dy\,ds
\,.
\end{align*}
The proof is then concluded by Lemma \ref{green}.
\end{proof}
\subsection{Effective dynamics} In this subsection, $\phi$, $\psi$ and  $f(t)$ are as in Theorems \ref{T1} and \ref{inhom}, $\delta_{0}$ denotes the Dirac delta distribution supported at the origin and $Mu(t)$ denotes the spherical mean of the continuous function $u$ over the sphere $\{x\in\RE^{3}:|x|=t\}$, i.e. (here $\sigma_{t}$ denotes the surface measure)
$$
Mu(t):=\frac1{4\pi t^{2}}\int_{|x|=t} u(x)\,d\sigma_{t}(x)\,;
$$
$u_{\text{\rm free}}$ and $\widehat u_{\text{\rm free}}$ denote the solutions of the homogeneous and inhomogeneous Cauchy problem as in Theorems \ref{T1} and \ref{inhom} respectively.
\begin{theorem}\label{T-eff} Let  $u_{\ve}$ and $u_{\ve,\text{\rm eff}}$ be the solutions of the inhomogeneous Cauchy problems
$$\begin{cases}
(\ve^{-2}\chi_{\Omega_{\ve}}+\chi_{{\RE^{3}\backslash\Omega_{\ve}}})\partial_{tt} u_{\ve}(t)=\Delta u_{\ve}(t)+f(t)\\
u_{\ve}(0)=\phi\\
\partial_{t} u_{\ve}(0)=\psi\,,
\end{cases}
\begin{cases}
\partial_{tt} u_{\ve,\text{\rm eff}}(t)=\Delta u_{\ve,\text{\rm eff}}(t)+\ve(\ve^{2}-1)\,q(t)\,\delta_{0}\\
u_{\ve,\text{\rm eff}}(0)=\phi\\
\partial_{t} u_{\ve,\text{\rm eff}}(0)=\psi\,,
\end{cases}
$$
where $q(t):=\sum_{k=1}^{\infty}q_{k}(t)$ and, for any $k\ge 1$, 
$q_{k}(t)$ solves the Cauchy  problem 
\be\label{cp-qeff}
\begin{cases}
\lambda_{k}\,\ddot q_{k}(t)=-q_{k}(t)+|\langle e_{k},1\rangle|^{2}\,h(t)\\
q_{k}(0)=0\\
\dot q_{k}(0)=0
\end{cases}
\ee
with
$$
h(t):=\Delta (u_{\text{\rm free}}+\widehat u_{\text{\rm free}})(t,0)=\frac{d\ }{dt}\, \big(tM {\Delta\phi}(t)\big)+tM{\Delta\psi}(t)+\int_{0}^{t}(t-s)(M{\Delta f(s)})(t-s)\,ds
$$
and the $e_{n}$'s and the $\lambda_{n}$'s are the eigenvectors and the eigenvalues of $N_{0}$.
Then, there exist $\tau\in(0,\frac1{11})$ and $\tau''\in(0,\frac12)$ such that 
$$
\sup_{0<t\le 1/\ve^{\tau}}\|u_{\ve}(t)-u_{\ve,\text{\rm eff}}(t)\|\lesssim\ve^{1+\tau''}\,,\qquad \ve\ll 1\,.
$$
\end{theorem}
\begin{proof} By Theorems \ref{T1} and \ref{inhom}, it suffices to show that, for any $k\ge 1$, 
$$q_{k}(t):=q_{1,k}(t)+q_{k,2}(t)+q_{k,3}(t)\,,
$$ solves \eqref{cp-qeff}, where
$$
q_{1,k}(t):=|\langle e_{k},1\rangle|^{2}
\int_{|y|<t}
\bigg(1-\cos\bigg(\frac{t-|y|}{{\lambda_{k}^{1/2}}}\bigg)\bigg)
\frac{\Delta^{2}\phi(y)}{4\pi|y|}\,dy
$$
$$
q_{k,2}(t):=|\langle e_{k},1\rangle|^{2}\,\frac1{\lambda_{k}^{1/2}}
\int_{|y|<t}\sin\bigg(\frac{t-|y|}{{\lambda_{k}^{1/2}}}\bigg)
\frac{\Delta\psi(y)}{4\pi|y|}\,dy\,,
$$ 
$$
q_{k,3}(t):=|\langle e_{k},1\rangle|^{2}\,\frac1{
\lambda_{k}^{1/2}}
\int_{0}^{t}\!\!\int_{|y|<t-s}\!\!\sin\bigg(\frac{t-s-|y|}{{\lambda_{k}^{1/2}}}\bigg)
\frac{\Delta f(s,y)}{4\pi|y|}\,dy\,ds\,.
$$
One has, 
\begin{align*}
\ddot q_{1,k}(t)&\,=|\langle e_{k},1\rangle|^{2}\,\frac1{\lambda^{1/2}_{k}}
\frac{d\ }{dt}\int_{|y|<t}\sin\bigg(\frac{t-|y|}{{\lambda_{k}^{1/2}}}\bigg)
\frac{\Delta^{2}\phi(y)}{4\pi|y|}\,dy\\
&\,=|\langle e_{k},1\rangle|^{2}\,\frac1{\lambda_{k}}\int_{|y|<t}\cos\bigg(\frac{t-|y|}{{\lambda_{k}^{1/2}}}\bigg)
\frac{\Delta^{2}\phi(y)}{4\pi|y|}\,dy\\
&\,=-\frac1{\lambda_{k}}\left(q_{1,k}(t)-|\langle e_{k},1\rangle|^{2}\int_{|y|<t}\frac{\Delta^{2}\phi(y)}{4\pi|y|}\,dy\right)\,,
\end{align*}
and
\begin{align*}
\ddot q_{2,k}(t)&\,=|\langle e_{k},1\rangle|^{2}\,\frac1{\lambda_{k}}
\frac{d\ }{dt}\int_{|y|<t}\cos\bigg(\frac{t-|y|}{{\lambda_{k}^{1/2}}}\bigg)
\frac{\Delta\psi(y)}{4\pi|y|}\,dy\\
&\,=|\langle e_{k},1\rangle|^{2}\,\frac1{\lambda_{k}}\left( tM\Delta\psi(t)-\frac1{\lambda_{k}^{1/2}}\int_{|y|<t}\cos\bigg(\frac{t-|y|}{{\lambda_{k}^{1/2}}}\bigg)
\frac{\Delta^{2}\phi(y)}{4\pi|y|}\,dy\right)\\
&\,=-\frac1{\lambda_{k}}\left(\,q_{2,k}(t)-|\langle e_{k},1\rangle|^{2}
tM{\Delta\psi}(t)\right)\,,
\end{align*}
\begin{align*}
&\,\ddot q_{3,k}(t)=|\langle e_{k},1\rangle|^{2}\,\frac1{
\lambda_{k}}\,\frac{d\ }{dt}
\int_{0}^{t}\!\!\int_{|y|<t-s}\!\!\cos\bigg(\frac{t-s-|y|}{{\lambda_{k}^{1/2}}}\bigg)
\frac{\Delta f(s,y)}{4\pi|y|}\,dy\,ds\\
&\,=|\langle e_{k},1\rangle|^{2}\,\frac1{\lambda_{k}}\int_{0}^{t}\left( (t-s)(M\Delta f(s))(t-s)-\frac1{\lambda_{k}^{1/2}}\int_{|y|<t-s}\sin\bigg(\frac{t-s-|y|}{{\lambda_{k}^{1/2}}}\bigg)
\frac{\Delta^{2} f(s,y)}{4\pi|y|}\,dy\right)\\
&\,=-\frac1{\lambda_{k}}\left(\,q_{2,k}(t)-|\langle e_{k},1\rangle|^{2}
\int_{0}^{t}(t-s)M({\Delta f(s)})(t-s)\,ds\right)\,.
\end{align*}
Furthermore, by Green's identity and by $0\notin\text{supp}(\phi)\subseteq{\RE^{3}\backslash\Omega_{\ve}}$,
\begin{align*}
\int_{|y|<t}\frac{\Delta^{2}\phi(y)}{4\pi|y|}\,dy=
\frac1{4\pi t^{2}}\int_{|y|=t}x\!\cdot\!\nabla\Delta\phi(y)\,d\sigma_{t}(y)+
\frac1{4\pi t^{2}}\int_{|y|=t}\Delta\phi(y)\,d\sigma_{t}(y)=\frac{d\ }{dt}\, tM{\Delta\phi}(t)\,.
\end{align*}
Hence, by Kirchhoff's formula, one gets 
$$\frac{d\ }{dt}\, \big(tM{\Delta\phi}(t)\big)+tM{\Delta\psi}(t)+
\int_{0}^{t}(t-s)(M{\Delta f(s)})(t-s)\,ds=u^{\Delta}_{\text{\rm free}}(t,0)\,,
$$ where $u^{\Delta}_{\text{\rm free}}$ denotes the solution of the Cauchy problem for the free wave equation with initial data $\Delta\phi$, $\Delta\psi$ and source $\Delta f(t)$; such a solution coincides with $\Delta (u_{\text{\rm free}}+\widehat u_{\text{\rm free}})$.
\end{proof}
\section{Appendix. Auxiliary results.}
\begin{lemma}\label{I1}
$$
\|I_{1}(t,\ve)\|\lesssim\ve^{1/2}\,\frac{e^{tc}}{c^{2n+3/2}}\ \|\Delta^{n+1}\psi\|\,.
$$
\end{lemma}
\begin{proof}
By Lemmata \ref{LN} and \ref{LG},
%\vfill\eject
\begin{align*}
&\|I_{1}(t,\ve)\|\\
\lesssim&\, e^{tc}\int_{0}^{c_{1}(c/\ve)^{1/4}}\!\!\!\|G_{c+i\gamma}\|\,\|(1+(c+i\gamma)^{2}N_{0})^{-1}\|\,\|G_{c+i\gamma}^{(1)}(\ve)\|\,\frac{d\gamma}{(c^{2}+\gamma^{2})^{n}}\ \|\Delta^{n+1}\psi\|\\
\lesssim&\,  \ve^{1/2}\,\frac{e^{tc}}{c^{1/2}}\bigg(\int_{0}^{c}\left(1+\frac1{c^{2}+\gamma^{2}}\right)\,\frac{d\gamma}{(c^{2}+\gamma^{2})^{n}}\\
&+\int_{c}^{c_{1}(c/\ve)^{1/4}}\frac{c^{2}+\gamma^{2}}{c\gamma}\left(1+(1+c^{2}+\gamma^{2})\,\frac{1}{c\gamma}\right)\,\frac{d\gamma}{(c^{2}+\gamma^{2})^{n}}\bigg) \|\Delta^{n+1}\psi\|
\\
\end{align*}%
\begin{align*}%
\lesssim&\,  \ve^{1/2}\,\frac{e^{tc}}{c^{1/2}}\bigg(\int_{0}^{c}\frac{d\gamma}{(c^{2}+\gamma^{2})^{n}}+\int_{0}^{c}\frac{d\gamma}{(c^{2}+\gamma^{2})^{n+1}}\\
&+\frac1{c^{2}}\int_{c}^{c_{1}(c/\ve)^{1/4}}\bigg(\frac{c}{\gamma}+\frac1{\gamma^{2}}+\frac{c^{2}}{\gamma^{2}}+1\bigg)\,\frac{d\gamma}{(c^{2}+\gamma^{2})^{n-1}}\bigg) \|\Delta^{n+1}\psi\|\\
%\end{align*}
%\begin{align*}
\lesssim&\, \ve^{1/2}\,\frac{e^{tc}}{c^{1/2}}\bigg(\frac1{c^{2n-1}}\int_{0}^{1}\frac{d\gamma}{(1+\gamma^{2})^{n}}+\frac1{c^{2n+1}}\int_{0}^{1}\frac{d\gamma}{(1+\gamma^{2})^{n+1}}\\%%%%
&+\frac1{c^{2n-1}}\int_{1}^{+\infty}\bigg(\frac1{\gamma}+\frac{1}{c^{2}\gamma^{2}}+\frac{1}{\gamma^{2}}+1\bigg)\,\frac{d\gamma}{(1+\gamma^{2})^{2n-1}}\bigg) \|\Delta^{n+1}\psi\|\\
\lesssim&\,\ve^{1/2}\,\frac{e^{tc}}{c^{1/2}}\bigg(\frac1{c^{2n-1}}+\frac1{c^{2n+1}}\bigg)\|\Delta^{n+1}\psi\|\\
\lesssim&\,\ve^{1/2}\,\frac{e^{tc}}{c^{2n+3/2}}\ \|\Delta^{n+1}\psi\|
\end{align*}
\end{proof}
%\vfill\eject
\begin{lemma}\label{I2}
$$
 \|I_{2}(t,\ve)\|\lesssim \ve^{1/2}\,\frac{e^{tc}}{c^{2n+1}}\ \|\Delta^{n+1}\psi\|\,.
$$
\end{lemma}
\begin{proof}
By Lemmata \ref{LL}, \ref{LL1}, \ref{CC} and \ref{LG} (with $s=5/4$),  
%\vfill\eject
\begin{align*}
&\|I_{2}(t,\ve)\|\\
\lesssim&\, e^{tc}\int_{0}^{c_{1}(c/\ve)^{1/4}}\!\!\!\|G_{c+i\gamma}(\ve)\|^{2}\|\Lambda_{c+i\gamma}(\ve)\|\,\frac{d\gamma}{(c^{2}+\gamma^{2})^{n}}\ \|\Delta^{n+1}\psi\|\\
\lesssim &\,\ve^{1/2}\,e^{tc}\bigg(\int_{0}^{c}\bigg(1+\frac{1}{c^{2}+\gamma^{2}}\bigg)^{\!\!5/4}\!\!\frac{c^{2}(1+c)}{(c^{2}+\gamma^{2})^{3/4}}\,
\frac{d\gamma}{(c^{2}+\gamma^{2})^{n}}\\
%\end{align*}
%\begin{align*}
&+\int_{c}^{c_{1}(c/\ve)^{1/4}}\!\!\!\frac{(1+(c+\gamma)^{2})^{5/4}}{(c\gamma)^{2}}\,(1+(c^{2}+\gamma^{2})^{1/2})\,\frac{(c^{2}+\gamma^{2})^{3}}{(c\gamma)^{2}}\,\frac{d\gamma}{(c^{2}+\gamma^{2})^{n}}
\bigg)\|\Delta^{n+1}\psi\|\\
\lesssim &\,\ve^{1/2}\,e^{tc}\bigg(\int_{0}^{c}\bigg(1+\frac{1}{c^{2}+\gamma^{2}}\bigg)^{\!\!5/4}\frac{c^{2\,}d\gamma}{(c^{2}+\gamma^{2})^{n+3/4}}
\\
&+\int_{c}^{c_{1}(c/\ve)^{1/4}}(1+(c+\gamma)^{2})^{5/4}\,
\frac{1+(c^{2}+\gamma^{2})^{1/2}}{(c\gamma)^{4}}
\,\frac{d\gamma}{(c^{2}+\gamma^{2})^{n-3}}\bigg)\|\Delta^{n+1}\psi\|\\
%\end{align*}
%\begin{align*}
\lesssim &\,\ve^{1/2}\,e^{tc}\bigg(\frac1{c^{2n-3/2}}\int_{0}^{1}\bigg(c^{2}+\frac{1}{1+\gamma^{2}}\bigg)^{\!\!5/4}\frac{d\gamma}{(1+\gamma^{2})^{n+3/4}}\\
&+\frac1{c^{2n+1}}\int_{1}^{+\infty}
(1+c^{2}(1+\gamma)^{2})^{5/4}\,
\frac{1+c(1+\gamma^{2})^{1/2}}{\gamma^{4}}
\,\frac{d\gamma}{(1+\gamma^{2})^{n-3}}
\bigg)\|\Delta^{n+1}\psi\|\\
\lesssim &\,\ve^{1/2}\,e^{tc}\bigg(\frac1{c^{2n-3/2}}+\frac1{c^{2n+1}}\bigg)\|\Delta^{n+1}\psi\|\\
\lesssim &\,\ve^{1/2}\,\frac{e^{tc}}{c^{2n+1}}\ \|\Delta^{n+1}\psi\|\,.
\end{align*}
\end{proof}
\begin{lemma}\label{green} Let $\varphi\in H^{6}(\RE^{3})$ be supported outside $\Omega_{\ve}$. Then
$$
\int_{|x|<t}\!\!\left(\frac{(t-|x|)^{2}}{2\lambda}-1+\cos\bigg(\frac{t-|x|}{\lambda^{1/2}}\bigg)\right){\Delta^{3}\varphi(x)}\,\frac{dx}{|x|}=\frac1\lambda
\int_{|x|<t}\!\!\left(1-\cos\bigg(\frac{t-|x|}{\lambda^{1/2}}\bigg)\right)\Delta^{2}\varphi(x)\,\frac{dx}{|x|}
$$
and
\begin{align*}
&\int_{|x|<t}\left(\frac{(t-|x|)^{3}}{6\lambda^{3/2}}-\frac{t-|x|}{\lambda^{1/2}}+\sin\bigg(\frac{t-|x|}{\lambda^{1/2}}\right)\bigg){\Delta^{3}\varphi(x)}\,\frac{dx}{|x|}
=\frac1{\lambda^{2}}
\int_{|x|<t}\sin\bigg(\frac{t-|x|}{\lambda^{1/2}}\bigg)\Delta\varphi(x)\,\frac{dx}{|x|}\,.
\end{align*}
\end{lemma}
\begin{proof} For any $g\in{\mathcal C}^{2}(\RE_{+})$ and for any $x\not=0$, one has
$$
\nabla\left(\frac{g(|x|)}{|x|}\right)
=\left(g'(|x|)-\,\frac{g(|x|)}{|x|}\right)\frac{x}{|x|^{2}}\,,\qquad
\Delta\left(\frac{g(|x|)}{|x|}\right)=\frac{g''(|x|)}{|x|}\,.
$$
Hence, by Green's formula and by $0\notin\text{supp}(\varphi)\subseteq{\RE^3\backslash\Omega_{\ve}}$,
\begin{align}\label{g}
&\int_{|x|<t}\frac{g(|x|)}{|x|}\,{\Delta^{n}\varphi(x)}\,dx=
\int_{|x|<t}\Delta\left(\frac{g(|x|)}{|x|}\right)\,\Delta^{n-1}\varphi(x)\,dx\nonumber\\
&+\frac{g(t)}{t^{2}}\int_{|x|=t}x\!\cdot\!\nabla\Delta^{n-1}\varphi(x)\,d\sigma_{t}(x)-\frac1t\int_{|x|=t}x\!\cdot\!\nabla\left(\frac{g(|x|)}{|x|}\right)\Delta^{n-1}\varphi(x)\,d\sigma_{t}(x)\nonumber\\
\,&=
\int_{|x|<t}\frac{g''(|x|)}{|x|}\,\Delta^{n-1}\varphi(x)\,dx+\frac{g(t)}{t^{2}}\int_{|x|=t}x\!\cdot\!\nabla\Delta^{n-1}\varphi(x)\,d\sigma_{t}(x)\\
\,&-\frac1t\left(g'(t)-\,\frac{g(t)}{t}\right)\int_{|x|=t}\Delta^{n-1}\varphi(x)\,d\sigma_{t}(x)\nonumber\,.
\end{align} 
Taking $g(s)=\frac{(t-s)^{2}}{2\lambda}-1+\cos\left(\frac{t-s}{\lambda^{1/2}}\right)$, one gets, by \eqref{g} and by $g(t)=g'(t)=0$
\begin{align*}%
&\int_{|x|<t}\left(\frac{(t-|x|)^{2}}{2\lambda}-1+\cos\bigg(\frac{t-|x|}{\lambda^{1/2}}\bigg)\right){\Delta^{3}\varphi(x)}\,\frac{dx}{|x|}\\
\,&=\frac1\lambda
\int_{|x|<t}\left(1-\cos\bigg(\frac{t-|x|}{\lambda^{1/2}}\bigg)\right)\Delta^{2}\varphi(x)\,\frac{dx}{|x|}\,.
\end{align*} 
Taking $g(s)=\frac{(t-s)^{3}}{6\lambda^{3/2}}-\frac{t-s}{\lambda^{1/2}}+\sin\left(\frac{t-s}{\lambda^{1/2}}\right)$, one gets, by using \eqref{g} twice and by $g(t)=g'(t)=g''(t)=g'''(t)=0$,
\begin{align*}%
&\int_{|x|<t}\left(\frac{(t-|x|)^{3}}{6\lambda^{3/2}}-\frac{t-|x|}{\lambda^{1/2}}+\sin\bigg(\frac{t-|x|}{\lambda^{1/2}}\bigg)\right){\Delta^{3}\varphi(x)}\,\frac{dx}{|x|}\\
\,&=
\int_{|x|<t|}\frac{g''(|x|)}{|x|}\,\Delta^{2}\varphi(x)\,dx=
\int_{|x|<t|}\frac{g''''(|x|)}{|x|}\,\Delta\varphi(x)\,dx\\
\,&=\frac1{\lambda^{2}}
\int_{|x|<t}\sin\bigg(\frac{t-|x|}{\lambda^{1/2}}\bigg)\Delta\varphi(x)\,\frac{dx}{|x|}\,.
\end{align*} 
\end{proof}
%\vskip30pt\noindent
\vfill\eject
\noindent
{\bf Acknowledgements.} A.P. acknowledges the support of the Next Generation EU-Prin project 2022 ''Singular Interactions and Effective Models in Mathematical Physics''  and of the National Group of Mathematical Physics (GNFM-INdAM). Part of this work was done while A.P. was visiting the Laboratoire de Math\'{e}matiques CNRS de Reims; this stay was financially supported by the Laboratoire Ypatia de Sciences Math\'{e}matiques. 

\vskip30pt\noindent
{\bf Declarations.}
The authors have no relevant financial or non-financial interests to disclose.
The authors have no competing interests to declare that are relevant to the content of this article.
Data sharing not applicable to this article as no data sets were generated or analyzed during the current study.

\end{document}